\documentclass[sigconf,edbt]{acmart-edbt2025}

\usepackage{subfigure}
\usepackage{url}
\usepackage{comment}
\usepackage{wrapfig,xcolor,comment}
\usepackage{bm}
\usepackage{amsthm}
\usepackage[ruled, vlined, linesnumbered]{algorithm2e}
\usepackage{multirow}
\usepackage{natbib}
\usepackage{float}
\newcommand{\tabincell}[2]{\begin{tabular}{@{}#1@{}}#2\end{tabular}}

\def\BibTeX{{\rm B\kern-.05em{\sc i\kern-.025em b}\kern-.08em
    T\kern-.1667em\lower.7ex\hbox{E}\kern-.125emX}}

\usepackage{booktabs} 

\definecolor{green}{rgb}{0.0, 0.55, 0.13}

\setcopyright{rightsretained}

\acmDOI{}

\acmISBN{978-3-98318-097-4}

\acmConference[EDBT 2025]{28th International Conference on Extending Database Technology (EDBT)}{25th March-28th March, 2025}{Barcelona, Spain}
\acmYear{2025}

\begin{document}
\title{Efficient Enumeration of Large Maximal $k$-Plexes}

\author{Qihao Cheng}
\authornote{The first three authors contributed equally to this work.}
\affiliation{%
  \institution{Tsinghua University}
  \country{} 
}
\email{cqh22@mails.tsinghua.edu.cn}

\author{Da Yan}
\authornotemark[1]
\affiliation{%
  \institution{Indiana University Bloomington}
  \country{} 
}
\email{yanda@iu.edu}

\author{Tianhao Wu}
\authornotemark[1]
\affiliation{%
  \institution{Tsinghua University}
  \country{} 
}
\email{wuth20@mails.tsinghua.edu.cn}

\author{Lyuheng Yuan}
\affiliation{%
  \institution{Indiana University Bloomington}
  \country{} 
}
\email{lyyuan@iu.edu}

\author{Ji Cheng}
\affiliation{%
  \institution{HKUST}
  \country{} 
}
\email{jchengac@connect.ust.hk}

\author{Zhongyi Huang}
\affiliation{%
  \institution{Tsinghua University}
  \country{} 
}
\email{zhongyih@tsinghua.edu.cn}

\author{Yang Zhou}
\affiliation{%
  \institution{Auburn University}
  \country{} 
}
\email{yangzhou@auburn.edu}

\renewcommand{\shortauthors}{}

\begin{abstract}
  Finding cohesive subgraphs in a large graph has many important applications, such as community detection and biological network analysis. Clique is often a too strict cohesive structure since communities or biological modules rarely form as cliques for various reasons such as data noise. Therefore, $k$-plex is introduced as a popular clique relaxation, which is a graph where every vertex is adjacent to all but at most $k$ vertices. In this paper, we propose a fast branch-and-bound algorithm as well as its task-based parallel version to enumerate all maximal $k$-plexes with at least $q$ vertices. Our algorithm adopts an effective search space partitioning approach that provides a lower time complexity, a new pivot vertex selection method that reduces candidate vertex size, an effective upper-bounding technique to prune useless branches, and three novel pruning techniques by vertex pairs. Our parallel algorithm uses a timeout mechanism to eliminate straggler tasks, and maximizes cache locality while ensuring load balancing. Extensive experiments show that compared with the state-of-the-art algorithms, our sequential and parallel algorithms enumerate large maximal $k$-plexes with up to $5 \times$ and $18.9 \times$ speedup, respectively. Ablation results also demonstrate that our pruning techniques bring up to $7 \times$ speedup compared with our basic algorithm. Our code is released at \url{https://github.com/chengqihao/Maximal-kPlex}.
\end{abstract}

%
%



\maketitle

\section{Introduction}
Finding cohesive subgraphs in a large graph is useful in various applications, such as finding protein complexes or biologically relevant functional groups~\cite{bader2003automated,bu2003topological,hu2005mining,ucar2006improving}, and social communities~\cite{li2014uncovering,hopcroft2004tracking} that can correspond to cybercriminals~\cite{weiss2015tracking}, botnets~\cite{ecrime,weiss2015tracking} and spam/phishing email sources~\cite{sac,sheng2009empirical}. One classic notion of cohesive subgraph is {\em clique} which requires every pair of distinct vertices to be connected by an edge. However, in real graphs, communities rarely appear in the form of cliques due to various reasons such as the existence of data noise~\cite{pattillo2013clique,www22maximal,d2k,kdd17}.

As a relaxed clique model, $k$-plex was first introduced in~\cite{seidman1978graph}, which is a graph where every vertex is adjacent to all but at most $k$ vertices. It has found extensive applications in the analysis of social networks~\cite{seidman1978graph}, especially in the community detection~\cite{pattillo2013clique,d2k}. However, mining $k$-plexes is NP-hard~\cite{DBLP:journals/jcss/LewisY80,DBLP:journals/ior/BalasundaramBH11}, so existing algorithms rely on branch-and-bound search which runs in exponential time in the worst case. Many recent works have studied the branch-and-bound algorithms for mining maximal $k$-plexes~\cite{d2k,aaai2020maximal,cikm22maximal,www22maximal} and finding a maximum $k$-plex~\cite{aaai17maximum,gao18ijcai,jiang21ijcai,colorBound,vldb22maximum}, with various techniques proposed to prune the search space. We will review these works in Section~\ref{sec:related}.

In this paper, we study the problem of enumerating all maximal $k$-plexes with at least $q$ vertices (and are hence large and statistically significant), and propose a more efficient branch-and-bound algorithm with new search-space pruning techniques and parallelization techniques to speed up the computation. Our search algorithm treats each set-enumeration subtree as an independent task, so that different tasks can be processed in parallel. Our main contributions are as follows:
\begin{itemize}
\item We propose a method for search space partitioning to create independent searching tasks, and show that its time complexity is $O\left(nr_1^kr_2\gamma_k^{D}\right)$, where $n$ is the number of vertices, $D$ is the graph degeneracy, and let $\Delta$ be the maximum degree, then $r_1=\min\left\{\frac{D\Delta}{q-2k+2},n\right\}$, $r_2 = \min\left\{\frac{D\Delta^2}{q-2k+2},nD\right\}$, and $\gamma_k<2$ is a constant close to 2.

\item We propose a new approach to selecting a pivot vertex to expand the current $k$-plex by maximizing the number of saturated vertices (i.e., those vertices whose degree is the minimum allowed to form a valid $k$-plex) in the $k$-plex. This approach effectively reduces the number of candidate vertices to expand the current $k$-plex.

\item We design an effective upper bound on the maximum size of any $k$-plex that can be expanded from the current $k$-plex $P$, so that if this upper bound is less than the user-specified size threshold, then the entire search branch originated from $P$ can be pruned.

\item We propose three novel effective pruning techniques by vertex pairs, and they are integrated into our algorithm to further prune the search space.

\item We propose a task-based parallel computing approach over our algorithm to achieve ideal speedup, integrated with a timeout mechanism to eliminate straggler tasks.

\item We conduct comprehensive experiments to verify the effectiveness of our techniques, and to demonstrate our superior performance over existing solutions.
\end{itemize}

\vspace{-1mm}
The rest of this paper is organized as follows. Section~\ref{sec:related} reviews the related work, and Section~\ref{sec:def} defines our problem and presents some basic properties of $k$-plexes. Then, Section~\ref{sec:algo} describes the branch-and-bound framework of our mining algorithm, Section~\ref{sec:prune} further describes the search space pruning techniques, and Section~\ref{sec:parallel} presents our task-based parallelization approach. Finally, Section~\ref{sec:results} reports our experiments, and Section~\ref{sec:conclude} concludes this paper.



\section{Related Work}\label{sec:related}
\noindent{\bf Maximal $k$-Plex Finding.} The Bron-Kerbosch (BK) algorithm~\cite{bk} is a backtracking search algorithm to enumerate maximal cliques, and can be extended to enumerate maximal $k$-plexes (see Section~\ref{sec:algo} for details). Many BK-style algorithms are proposed with various effective search space pruning techniques. Specifically, D2K~\cite{d2k} proposes a simple pivoting technique to cut useless branches, which generalizes the pivoting technique for maximal clique finding. A more effective pivoting technique is found by FaPlexen~\cite{aaai2020maximal}. More recently, FP~\cite{cikm22maximal} adopts a new pivoting technique and uses upper-bound-based pruning for maximal $k$-plex enumeration, but the time complexity is still the same as previous works~\cite{d2k, aaai2020maximal}, which is improved by our current work. ListPlex~\cite{www22maximal} adopts the sub-tasking scheme that partitions the search space efficiently, but it uses the less effective pivoting and branching schemes of FaPlexen, which is avoided in our current work. None of the works has considered effective vertex-pairs pruning techniques proposed in this paper.

Besides BK, maximal $k$-plexes can also enumerated by a reverse search framework~\cite{sigmod15}. The key insight is that given a valid $k$-plex $P$, it is possible to find another valid one by excluding some existing vertices from and including some new ones to $P$. Starting from an initial solution, \cite{sigmod15} conducts DFS over the solution graph to enumerate all solutions. While the algorithm provides a polynomial delay (i.e., time of waiting for the next valid $k$-plex) so that it is guaranteed to find some solutions in bounded time, it is less efficient than BK when the goal is to enumerate all maximal $k$-plexes. Reverse search has also been adapted to work with bipartite graphs~\cite{maximal_bi}.


\vspace{1mm}
\noindent{\bf Maximum $k$-Plex Finding.} Conte et al.~\cite{kdd17} notice that any node of a $k$-plex $P$ with $|P|\geq q$ is included in a clique of size at least $\lceil q/k\rceil$, which is used to prune invalid nodes. However, it is necessary to enumerate all maximal cliques which is expensive per se. To find a maximum $k$-plex, \cite{kdd17} uses binary search to guess the maximum $k$-plex size for vertex pruning, and then mines $k$-plexes on the pruned graph to see if such a maximum $k$-plex can be found, and if not, the maximum $k$-plex size threshold is properly adjusted for another round of search. However, this approach may fail for multiple iterations before finding a maximum $k$-plex, so is less efficient than the branch-and-bound algorithms.

BS~\cite{aaai17maximum} pioneers a number of pruning techniques in the brand-and-bound framework for finding a maximum $k$-plex, including the pivoting technique of FaPlexen~\cite{aaai2020maximal}. In maximum $k$-plex finding, if the current maximum $k$-plex is $P$, then we can prune any branch that cannot generate a $k$-plex with at least $|P|+1$ vertices (i.e., upper bound $\leq |P|+1$). BnB~\cite{gao18ijcai} proposes upper bounds and pruning techniques based on deep structural analysis, KpLeX~\cite{jiang21ijcai} proposes an upper bound based on vertex partitioning, and Maplex~\cite{colorBound} proposes an upper bound based on graph coloring which is later improved by RGB~\cite{rgb}. kPlexS~\cite{vldb22maximum} proposes a CTCP technique to prune the vertices and edges using the second-order property, 
and it shows that the reduced graph by CTCP is guaranteed to be no larger than that computed by BnB, Maplex and KpLeX. kPlexS also proposed new techniques for branching and pruning, 
and outperforms BnB, Maplex and KpLeX. In~\cite{DBLP:conf/ijcai/WangZLX23, DBLP:conf/ijcai/JiangXZWZ23}, to find maximum $k$-plexes, an algorithm for the $d$-BDD problem is applied and a refined upper bound is proposed. 

\vspace{1mm}
\noindent{\bf Other Dense Subgraphs.} There are other definitions of dense subgraphs. Specifically, \cite{DBLP:conf/approx/Charikar00} finds subgraphs to maximize the average degree~\cite{DBLP:conf/approx/Charikar00} solvable by a flow-based algorithm,  \cite{DBLP:journals/algorithmica/FeigePK01} finds the $k$-vertex subgraph with the most edges, \cite{DBLP:journals/dam/AsahiroHI02} finds $k$-vertex subgraphs with at least $f(k)$ edges for an edge-density function $f(.)$, and \cite{DBLP:conf/kdd/TsourakakisBGGT13} proposes a density measure based on edge surplus to extract a higher-quality subgraph called optimal quasi-clique. However, those problems are very expensive and solved by approximate algorithms, while we target exact $k$-plex solutions.

Besides $k$-plex, $\gamma$-quasi-clique is the other popular type of clique relaxation whose exact algorithms gained a lot of attention. Branch-and-bound algorithms Crochet~\cite{Pei05,Pei09}, Cocain~\cite{cocain}, and Quick~\cite{quick} mine maximal $\gamma$-quasi-cliques exactly, and parallel and distributed algorithms have also been developed by our prior works~\cite{DBLP:journals/pvldb/Guo0O0K20,DBLP:journals/vldb/KhalilYGY22,DBLP:conf/icde/GuoYYKLJZ22}. Unlike $k$-plex where the restriction at each vertex is on the absolute number of missing edges allowed, $\gamma$-quasi-clique places this restriction on the ratio of missing edges (i.e., $(1-\alpha)$ fraction) at each vertex. This difference makes $\gamma$-quasi-clique not satisfying the hereditary property as in $k$-plexes and cliques~\cite{pattillo2013maximum}, making the BK algorithm not applicable, so more expensive branch-and-bound algorithms with sophisticated pruning rules to check are needed.

\renewcommand{\arraystretch}{1.3} 
\begin{table}[t]
 \scriptsize
\centering
  \caption{List of Important Notations}
  \label{tbl-notation}
  \resizebox{\columnwidth}{!}{
      \begin{tabular}{c|c}
        \toprule[2pt]
        {\bf Notation} & {\bf Description} \\
        \hline
        $P,~C$, and $X$  & the current $k$-plex, candidate set, and exclusive set\\
        \hline
        $S$             & a subset of $N^2_{G_i}(v_i)$ \\
        \hline
        $T_{v_i\cup S}$ &a sub-task for set-enumeration search\\
        \hline
        $\eta$ & the degeneracy ordering of $G$, $\{v_1, v_2, \ldots, v_n\}$ \\
        \hline
        $V_{<\eta}(v_i)$, and $V_{\ge\eta}(v_i)$ & $\{v_1, v_2, \ldots, v_{i-1}\}$, and $\{v_i, v_2, \ldots, v_n\}$\\
        \hline
        $G_i$ & the subgraph induced by vertices in $V_{\ge\eta}(v_i)$ within 2 hops from $v_i$\\
        \hline
        $P_S$, and $C_S$ &a sub-task with $P=\{v_i\}\cup S$, and its candidate set\\
        \hline
        $P_m$ & a maximum $k$-plex containing the current $k$-plex $P$\\
        \hline
        $sup_P(v)$ & the maximum \# of $v$'s non-neighbors outside $P$ that can be added to $P$\\
        \hline
        $ub(P)$ &the upper bound of the maximum k-plex that $P$ can extend to\\
        \bottomrule[2pt]
      \end{tabular}
}
\label{table::notation}
\end{table}

\section{Problem Definition}\label{sec:def}
For ease of presentation, we first define some basic notations. More notations will be defined in Sections~\ref{sec:algo} and~\ref{sec:prune} when describing our algorithm, and Table~\ref{table::notation} lists the important notations for quick lookup.

\vspace{1mm}
\noindent {\bf Notations.} 
We consider an undirected and unweighted simple graph $G=(V, E)$, where $V$ is the set of vertices, and $E$ is the set of edges. We let $n=|V|$ and $m=|E|$ be the number of vertices and the number of edges, respectively. The diameter of $G$, denoted by $\delta(G)$, is the shortest-path distance of the farthest pair of vertices in $G$, measured by the \# of hops.

For each vertex $v\in V$, we use $N^c_G(v)$ to denote the set of vertices with distance exactly $c$ to $v$ in $G$. For example, $N^1_G(v)$ is $v$'s direct neighbors in $G$, which we may also write as $N_G(v)$; and $N^2_G(v)$ is the set of all vertices in $G$ that are 2 hops away from $v$. The degree of a vertex $v$ is denoted by $d_G(v)=|N_G(v)|$, and the maximum vertex degree in $G$ is denoted by $\Delta$.

We also define the concept of {\em non-neighbor}: a vertex $u$ is a non-neighbor of $v$ in $G$ if $(u, v)\not\in E$. Accordingly, the set of non-neighbors of $v$ is denoted by $\overline{N_G}(v)=V-N_G(v)$, and we denote its cardinality by $\overline{d_G}(v)=|\overline{N_G}(v)|$.

Given a vertex subset $S\subseteq V$, we denote by $G[S]=(S, E[S])$ the subgraph of $G$ induced by $S$, where $E[S]=\{(u,v)\in E\,|\,u\in S\wedge v\in S\}$. We simplify the notation $N_{G[S]}(v)$ to $N_S(v)$, and define other notations such as $\overline{N_S}(v)$, $d_S(v)$, $\overline{d_S}(v)$ and $\delta(S)$ in a similar manner.

The $k$-core of an undirected graph $G$ is its largest induced subgraph with minimum degree $k$. The degeneracy of $G$, denoted by $D$, is the largest value of $k$ for which a $k$-core exists in $G$. The degeneracy of a graph may be computed in linear time by a peeling algorithm that repeatedly removes the vertex with the minimum current degree at a time~\cite{bz}, which produces a degeneracy ordering of vertices denoted by $\eta=[v_1, v_2, \ldots, v_n]$. All the consecutively removed vertices with the minimum current degree being $k$ ($k=0, 1, \cdots, D$) constitute a $k$-shell, and in degeneracy ordering, vertices are listed in segments of $k$-shells with increasing $k$. We order vertices in the same $k$-shell by vertex ID (from the input dataset) to make $\eta$ unique, though our tests by shuffling within-shell vertex ordering show that it has a negligible impact on the time difference for our $k$-plex mining. In a real graph, we usually have $D\ll n$. 

\vspace{1mm}
\noindent {\bf Problem Definition.} We next define our mining problem. As a relaxed clique model, a $k$-plex is a subgraph $G[P]$ that allows every vertex $u$ to miss at most $k$ links to vertices of $P$ (including $u$ itself), i.e., $d_P(u)\geq |P|-k$ (or, $\overline{d_P}(u)\leq k$):
\begin{definition}(\emph{$k$-Plex})\label{definition}
    Given an undirected graph $G=(V, E)$ and a positive integer $k$, a set of vertices $P\subseteq V$ is a $k$-plex iff for every $u\in P$, its degree in $G[P]$ is no less than $(|P|-k)$.
\end{definition}

Note that $k$-plex satisfies the hereditary property:
\begin{theorem}\emph{(Hereditariness)}
\label{lemma::hereditary}
   Given a $k$-plex $P\subseteq V$, any subset $P'\subseteq P$ is also a $k$-plex.
\end{theorem}

This is because for any $u\in P'$, we have $u\in P$ and since $P$ is a $k$-plex, $\overline{d_P}(u)=|\overline{N_P}(u)|\leq k$. Since $\overline{N_{P'}}(u)\subseteq \overline{N_P}(u)$, we have $\overline{d_{P'}}(u)=|\overline{N_{P'}}(u)|\leq k$, so $P'$ is also a $k$-plex.

Another important property is that if a $k$-plex $P$ satisfies $|P|>2k-c$, then $G[P]$ is connected with the diameter $\delta(P)\leq c$ ($c\geq 2$)~\cite{aaai17maximum}. A common assumption by existing works~\cite{d2k,cikm22maximal} is the special case when $c=2$:
\begin{theorem}
\label{lemma::diameter}
    Given a $k$-plex $P$, if  $|P|\geq 2k-1$, then $\delta(P)\leq 2$.
\end{theorem}
This is a reasonable assumption since natural communities that $k$-plexes aim to discover are connected, and we are usually interested in only large (hence statistically significant) $k$-plexes with size at least $q$. For $k\leq 5$, we only require $q\geq 2k-1=9$. Note that a $k$-plex with $|P|=2k-2$ may be disconnected, such as one formed by two disjoint $(k-1)$-cliques.

A $k$-plex is said to be maximal if it is not a subgraph of any larger $k$-plex. 
We next formally define our problem: 
\begin{definition}(\emph{Size-Constrained Maximal $k$-Plex Enumeration})
    \label{problem1}
   Given a graph $G=(V, E)$ and an integer size threshold $q\geq 2k-1$, find all the maximal $k$-plexes with at least $q$ vertices. 
\end{definition}

Note that instead of mining $G$ directly, we can shrink $G$ into its $(q-k)$-core for mining, which can be constructed in $O(m+n)$ time using the peeling algorithm that keeps removing those vertices with degree less than $(q-k)$:
\begin{theorem}
    \label{lemma::q-k core}
    Given a graph $G=(V, E)$, all the $k$-plexes with at least $q$ vertices must be contained in the $(q-k)$-core of $G$. 
\end{theorem}

This is because for any vertex $v$ in a $k$-plex $P$, $d_P(v)\geq |P|-k$, and since we require $|P|\geq q$, we have $d_P(v)\geq q-k$. 

\section{Branch-and-Bound Algorithm}\label{sec:algo}
This section describes the branch-and-bound framework of our mining algorithm. Section~\ref{sec:prune} will further describe the pruning techniques that we use to speed up our algorithm.

\begin{figure}[t]
\centering
\includegraphics[width=\columnwidth]{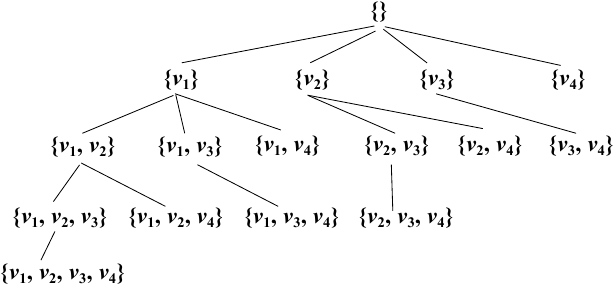}
\caption{Set-Enumeration Search Tree}\label{set_enum}
\end{figure}

\vspace{1mm}
\noindent {\bf Set-Enumeration Search.} Figure~\ref{set_enum} shows the set-enumeration tree $T$ for a graph $G$ with four vertices $V=\{v_1, v_2, v_3, v_4\}$ where we assume vertex order $v_1<v_2<v_3<v_4$. Each tree node represents a vertex set $P$, and only those vertices larger than the largest vertex in $P$ are used to extend $P$. For example, in Figure~\ref{set_enum}, node $\{v_1,v_3\}$ can be extended with $v_4$ but not $v_2$ since $v_2<v_3$; in fact, $\{v_1,v_2,v_3\}$ is obtained by extending $\{v_1,v_2\}$ with $v_3$.
Let us denote $T_P$ as the subtree of $T$ rooted at a node with set $P$. Then, $T_P$ represents a search space for all possible $k$-plexes that contain all the vertices in $P$.
We represent the task of mining $T_P$ as a pair $\langle P, C\rangle$, where $P$ is the set of vertices assumed to be already included, and $C\subseteq(V-P)$ keeps those vertices that can extend $P$ further into a valid $k$-plex. 
The task of mining $T_P$, i.e., $\langle P, C\rangle$, can be recursively decomposed into tasks that mine the subtrees rooted at the children of $P$ in $T_P$.

Algorithm~\ref{algo:bk} describes how this set-enumeration search process is generated, where we first ignore the red parts, and begin by calling \textit{bron\_kerbosch}$(P=\emptyset, C=V)$. 
Specifically, in each iteration of the for-loop in Lines~\ref{bk:line2}--\ref{bk:line7}, we consider the case where $v_i$ is included into $P$ (see $P'$ in Lines~\ref{bk:line3} and~\ref{bk:line6}). 
Here, Line~\ref{bk:line4} is due to the hereditary property: if $P'\cup\{u\}$ is not a $k$-plex, then any superset of $P'\cup\{u\}$ cannot be a $k$-plex, so $u\not\in C'$. 
Also, Line~\ref{bk:line3} removes $v_i$ from $C$ so in later iterations, $v_i$ is excluded from any subgraph grown from $P$.

Note that while the set-enumeration tree in Figure~\ref{set_enum} ensures no redundancy, i.e., every subset of $V$ will be visited at most once, it does not guarantee set maximality: even if $\{v_1, v_2, v_4\}$ is a $k$-plex, $\{v_2, v_4\}$ will still be visited but it is not maximal.

\begin{algorithm}[!t]
\DontPrintSemicolon
    {\color{red} {\bf if} $C=\emptyset$ {\bf and} $X=\emptyset$ {\bf do}\ \ \ \ \bf{output} $P$, {\bf return}}\label{bk:line1}\;
    \ForEach{vertex $v_i\in C$}{\label{bk:line2}
        $P'\gets P\cup\{v_i\}$, $C\gets C-\{v_i\}$\label{bk:line3}\;
        $C'\gets \{u\,|\,u\in C\ \wedge\ P'\cup\{u\}\mbox{ is a $k$-plex}\}$\label{bk:line4}\;
        {\color{red} $X'\gets \{u\,|\,u\in X\ \wedge\ P'\cup\{u\}\mbox{ is a $k$-plex}\}$}\label{bk:line5}\;
        \textit{bron\_kerbosch}$(P', C', {\color{red} X'})$\label{bk:line6}\;
        {\color{red} $X\gets X\cup\{v_i\}$}\label{bk:line7}\;
    }
\caption{\textit{bron\_kerbosch}$(P, C, {\color{red} X})$}\label{algo:bk}
\end{algorithm}

\vspace{1mm}
\noindent {\bf Bron-Kerbosch Algorithm.} The Bron-Kerbosch algorithm as shown in Algorithm~\ref{algo:bk} avoids outputting non-maximal $k$-plexes with the help of an exclusive set $X$. The algorithm was originally proposed to mine maximal cliques~\cite{bk}, and has been recently adapted for mining maximal $k$-plexes~\cite{aaai2020maximal,d2k}.

Specifically, after each iteration of the for-loop where since we consider the case with $v_i$ included into $P$, we add $v_i$ to $X$ in Line~\ref{bk:line7} so that in later iterations (where $v_i$ is not considered for extending $P$), $v_i$ will be used to check result maximality.

We can redefine the task of mining $T_P$ as a triple $\langle P, C, X\rangle$ with three disjoint sets, where the exclusive set $X$ keeps all those vertices that have been considered before (i.e., added by Line~\ref{bk:line7}), and can extend $P$ to obtain larger $k$-plexes (see Line~\ref{bk:line5} which refines $X$ into $X'$ based on $P'$). Those $k$-plexes should have been found before.

When there is no more candidate to grow $P$ (i.e., $C=\emptyset$ in Line~\ref{bk:line1}), if $X\neq\emptyset$, then based on Line~\ref{bk:line5}, $P\cup\{u\}$ is a $k$-plex for any $u\in X$, so $P$ is not maximal. Otherwise, $P$ is maximal (since such a $u$ does not exist) and outputted. 
For example, let $P=\{v_2, v_4\}$ and $X=\{v_1\}$, then we cannot output $P$ since $\{v_1, v_2, v_4\}\supseteq P$ is also a $k$-plex, so $P$ is not the maximal one.

\vspace{1mm}
\noindent {\bf Initial Tasks.} Referring to Figure~\ref{set_enum} again, the top-level tasks are given by $P=\{v_1\}$, $\{v_2\}$, $\{v_3\}$ and $\{v_4\}$, which are generated by \textit{bron\_kerbosch}$(P=\emptyset, C=V, X=\emptyset)$. It is common to choose the precomputed degeneracy ordering $\eta=[v_1, v_2, \ldots, v_n]$ to conduct the for-loop in Line~\ref{bk:line2}, which was found to generate more load-balanced tasks $T_{\{v_i\}}$~\cite{www22maximal,cikm22maximal,aaai2020maximal,d2k}. 
Intuitively, each vertex $v_i$ is connected to at most $D$ vertices among later candidates $\{v_{i+1}, v_{i+2}, \ldots, v_n\}$ based on the peeling process, and $D$ is usually a small value.

Note that we do not need to mine each $T_{\{v_i\}}$ over the entire $G$. 
Let us define $V_{<\eta}(v_i)=\{v_1, v_2, \ldots, v_{i-1}\}$ and $V_{\ge\eta}(v_i)=\{v_i, v_{i+1}, v_{i+2},$ $\ldots, v_n\}$, then we only need to mine $T_{\{v_i\}}$ over
\begin{equation}\label{eq:gi}
G_i=G\left[V_{\ge\eta}(v_i)\cap\left(\{v_i\}\cup N_G(v_i)\cup N^2_G(v_i)\right)\right],
\end{equation}
since candidates in $C$ must be after $v_i$ in $\eta$, and must be within two hops from $v_i$ according to Theorem~\ref{lemma::diameter}. In fact, since $G_i$ tends to be dense, it is efficient when $G_i$ is represented by an adjacency matrix~\cite{vldb22maximum}. 
We call $v_i$ as a seed vertex, and call $G_i$ as a seed subgraph.

\begin{figure}[t]
\centering
\includegraphics[width=0.64\columnwidth]{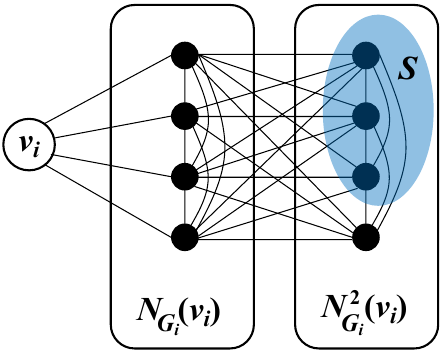}
\caption{Decomposition of Top-Level Task $T_{v_i}$}\label{N2}
\end{figure}

As a further optimization, we decompose $T_{\{v_i\}}$ into disjoint sub-tasks $T_{\{v_i\}\cup S}$ for subsets $S\subseteq N^2_{G_i}(v_i)$, where the vertices of $S$ are the only vertices in $N^2_{G_i}(v_i)$ allowed to appear in a $k$-plex found in $T_{\{v_i\}\cup S}$, and other candidates have to come from $N_{G_i}(v_i)$ (see Figure~\ref{N2}). We only need to consider $|S|<k$ since otherwise, $v_i$ has at least $k$ non-neighbors $S\subseteq N^2_{G_i}(v_i)$, plus $v_i$ itself, $v_i$ misses $(k+1)$ edges which violates the $k$-plex definition, so $\{v_i\}\cup S$ cannot be a $k$-plex, neither can its superset due to the hereditary property.

In summary, each search tree $T_{\{v_i\}}$ creates a {\em task group} sharing the same graph $G_i$ (c.f.\ Eq~(\ref{eq:gi})), where each task mines the search tree $T_{\{v_i\}\cup S}$ for a subset $S\subseteq N^2_{G_i}(v_i)$ with $|S|<k$.

Algorithm~\ref{alg::1} shows the pseudocode for creating initial task groups, where in Line~\ref{algo1:line1} we shrink $G$ into its $(q-k)$-core by Theorem~\ref{lemma::q-k core}, so $n$ is reduced. Line~\ref{algo1:line2} then orders the vertices of $G$ in degeneracy order to keep the size of all $G_i$ small to generate more load-balanced tasks by bounding the candidate size $|C|$. This ordering is also essential for our time complexity analysis in Section~\ref{sec:prune}.

\begin{algorithm}[!t]
\DontPrintSemicolon
    \KwIn{Graph $G=(V,E)$,\ \ $k$,\ \ $q\ge2k-1$}
    $G\gets$ the $(q-k)$-core of $G$\ \ \ \ \ \ \ \ \///\ Using Theorem~\ref{lemma::q-k core}\label{algo1:line1}\;
    $\eta=\{v_1,\dots, v_{n}\}$ is the degeneracy ordering of $V$\label{algo1:line2}\;
    \For{$i=1,2,\dots,n-q+1$}{\label{algo1:line3}
        $V_i\gets\{v_i,v_{i+1},\dots,v_{n}\}\cap\left(\{v_i\}\cup N_G(v_i)\cup N^2_G(v_i)\right)$\label{algo1:line4}\!\!\!\!\;
        $V_i'\leftarrow\{v_1,v_{2},\dots,v_{i-1}\}\cap\left(N_G(v_i)\cup N^2_G(v_i)\right)$\label{algo1:line5}\;
        $G_i\leftarrow G[V_i]$, and apply further pruning over $G_i$\label{algo1:line6}\;
        \ForEach{$S\subseteq N^2_{G_i}(v_i)$ that $|S|\le k-1$}{\label{algo1:line7}
    	   $P_S\leftarrow\{v_i\}\cup S$,\ \ \ \ $C_S\leftarrow N_{G_i}(v_i)$\label{algo1:line8}\;
           $X_S\leftarrow V_i'\cup(N^2_{G_i}(v_i)-S)$\label{algo1:line9}\;
           Call Branch$(G_i,k,q,P_S,C_S,X_S)$\label{algo1:line10}\;
        }
    }
\caption{Enumerating $k$-Plex with Initial Tasks}
\label{alg::1}
\end{algorithm}

We then generate initial task groups $T_{\{v_i\}}$ using the for-loop from Line~\ref{algo1:line3}, where we skip $i>n-q+1$ since $|V_i|<q$ in this case; here, $V_i$ is the vertex set of $G_i$ (see Lines~\ref{algo1:line4} and~\ref{algo1:line6}).
For each task $T_{\{v_i\}\cup S}=\langle P,C,X\rangle$ of the task group $T_{\{v_i\}}$, we have $P=\{v_i\}\cup S$ and $C\subseteq N_{G_i}(v_i)\triangleq C_S$ (Line~\ref{algo1:line8}). Let us define $N^{1,2}_{<\eta}(v_i)=V_{<\eta}(v_i)\cap\left(N_G(v_i)\cup N^2_G(v_i)\right)$ (i.e., $V_i'$ in Line~\ref{algo1:line5}), then we also have $X\subseteq N^{1,2}_{<\eta}(v_i)\,\cup\,(N^2_{G_i}(v_i)-S)\triangleq X_S$ (Line~\ref{algo1:line9}). This is because vertices of $(N^2_{G_i}(v_i)-S)$ may be considered by other sub-tasks $T_{\{v_i\}\cup S'}$, and vertices of $N^{1,2}_{<\eta}(v_i)$ may be considered by other task groups $T_{\{v_j\}}$ $(j<i)$ (if $v_j$ is more than 2 hops away from $v_i$, it cannot form a $k$-plex with $v_i$, so $v_j$ is excluded from $X_S$). 
Line~\ref{algo1:line7} of Algorithm~\ref{alg::1} is implemented by the set-enumeration search of $S$ over $N^2_G(v_i)$, similar to Algorithm~\ref{algo:bk} Lines~\ref{bk:line2}, \ref{bk:line3} and~\ref{bk:line6}.

Finally, to maintain the invariant of Bron-Kerbosch algorithm (c.f., Lines~\ref{bk:line4}--\ref{bk:line5} of Algorithm~\ref{algo:bk}), we set $C\gets \{u\,|\,u\in C_S\ \wedge\ P\cup\{u\}\mbox{ is a $k$-plex}\}$ and $X\gets \{u\,|\,u\in X_S\ \wedge\ P\cup\{u\}\mbox{ is a $k$-plex}\}$, and mine $T_{\{v_i\}\cup S}=\langle P,C,X\rangle$ recursively using the Bron-Kerbosch algorithm of Algorithm~\ref{algo:bk} over $G_i$. Instead of directly running Algorithm~\ref{algo:bk}, we actually run a variant to be described in Algorithm~\ref{alg::2} which applies more pruning techniques, and refines $C_S$ and $X_S$ into $C$ and $X$, respectively, at the very beginning. This branch-and-bound sub-procedure is called in Line~\ref{algo1:line10} of Algorithm~\ref{alg::1}.

\begin{algorithm}[t]
    \DontPrintSemicolon
    \SetKwBlock{Begin}{function}{end function}
    \Begin($\text{Branch}{(}G,k,q,P,C,X {)}$){
    $C\leftarrow\{v\,|\,v\in C\ \wedge\ P\cup\{v\} \text{ is a }k\text{-plex}\}$\label{alg2:line2}\;
    $X\leftarrow\{v\,|\,v\in X\ \wedge\ P\cup\{v\} \text{ is a }k\text{-plex}\}$\label{alg2:line3}\;
    \If{$C = \emptyset$}{\label{alg2:line4}
        {\bf if} $X=\emptyset$ {\bf and} $|P|\ge q$ {\bf then}\ \ \ \ Output $P$\label{alg2:line5}\;
        \Return\label{alg2:line6}\;
    }
    $M_0\gets$ the subset of $P\cup C$ with minimum degree in $G[P\cup C]$\label{alg2:line7}\;
    $M\gets$ the subset of $M_0$ with maximum $\overline{d_P}(v)$\label{alg2:line8}\;
    {\bf if} $M\cap P\neq \emptyset$ {\bf then}\ \ \ \ Pick a pivot $v_p\in M\cap P$\label{alg2:line9}\;
    {\bf else}\ \ \ \ Pick a pivot $v_p\in M\cap C$\label{alg2:line10}\;
    \If{$d_{P\cup C}(v_p)\ge |P|+|C|-k$}{\label{alg2:line11}
        \If{$P\cup C$ is a maximal $k$-plex}{\label{alg2:line12} 
            {\bf if} $|P\cup C|\ge q$ {\bf then}\ \ \ \ Output $P\cup C$\;
        }
        \Return\label{alg2:line14}\;
    }
    \If{$v_p\in P$}{\label{alg2:line15}
        Re-pick a pivot $v_{new}$ from $\overline{N_C}(v_p)$ using the same rules as in Lines~\ref{alg2:line7}--\ref{alg2:line10};\ \ \ \ 
        $v_p\leftarrow v_{new}$\label{alg2:line16}\;
    }
    Compute the upper bound $ub$ of the size of any $k$-plex that can be expanded from $P\cup\{v_p\}$\label{alg2:line17}\;
    \If{$ub\ge q$}{\label{alg2:line18}
        $\text{Branch}(G,k,q,P\cup\{v_p\},C-\{v_p\},X)$\label{alg2:line19}\;
    }
    $\text{Branch}(G,k,q,P,C-\{v_p\},X\cup\{v_p\})$\label{alg2:line20}\;
}
\caption{Branch-and-Bound Search}    
\label{alg::2}
\end{algorithm}

\vspace{1mm}
\noindent {\bf Branch-and-Bound Search.} Algorithm~\ref{alg::2} first updates $C$ and $X$ to ensure that each vertex in $C$ or $X$ can form a $k$-plex with $P$ (Lines~\ref{alg2:line2}--\ref{alg2:line3}). If $C=\emptyset$ (Line~\ref{alg2:line2}), there is no more candidate to expand $P$ with, so Line~\ref{alg2:line6} returns. Moreover, if $X=\emptyset$ (i.e., $P$ is maximal) and $|P|\ge q$, we output $P$  (Line~\ref{alg2:line5}).

Otherwise, we pick a pivot $v_p$ (Lines~\ref{alg2:line7}--\ref{alg2:line10} and~\ref{alg2:line15}--\ref{alg2:line16}) and compute an upper bound $ub$ of the maximum size of any $k$-plex that $P\cup\{v_p\}$ may expand to (Line~\ref{alg2:line17}). The branch expanding $P\cup\{v_p\}$ is filtered if $ub<q$ (Lines~\ref{alg2:line18}--\ref{alg2:line19}), while the branch excluding $v_p$ is always executed in Lines~\ref{alg2:line20}. We explain how $ub$ is computed later.

\vspace{1mm}
\noindent {\bf Pivot Selection.} We next explain our pivot selection strategy. 
Specifically, Lines~\ref{alg2:line7}--\ref{alg2:line10} select $v_p\in P\cup C$ to be a vertex with the minimum degree in $G[P\cup C]$, so that in Line~\ref{alg2:line11}, if $d_{P\cup C}(v_p)\ge |P|+|C|-k$, then for any other $v\in P\cup C$, we have $d_{P\cup C}(v)\ge d_{P\cup C}(v_p)\ge |P|+|C|-k$, and hence $P\cup C$ is a $k$-plex that we then examine for maximality. In this case, we do not need to expand further so Line~\ref{alg2:line14} returns. 
In Line~\ref{alg2:line12}, we check if $P\cup C$ is maximal by checking if $\{v\,|\,v\in X\ \wedge\ P\cup C\cup\{v\} \text{ is a }k\text{-plex}\}$ is empty.

Note that among those vertices with the minimum degree in $G[P\cup C]$, we choose $v_p$ with the maximum $\overline{d_P}(v)$ (Line~\ref{alg2:line8}) which tends to prune more candidates in $C$. Specifically, if $\overline{d_P}(v_p)=k$ and $v_p$ is in (or added to) $P$, then $v_p$'s non-neighbors in $C$ are pruned; such a vertex $v_p$ is called {\em saturated}.

\begin{figure}[t]
\centering
\includegraphics[width=0.65\columnwidth]{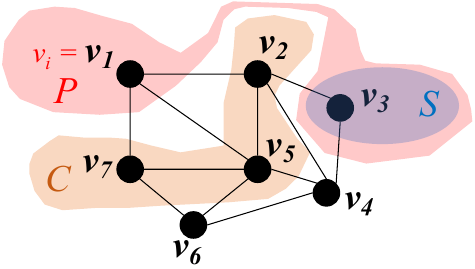}
\caption{A Toy Graph for Illustration}\label{fig::example_graph}
\end{figure}

Note that \textit{if more saturated vertices are included in $P$, more vertices in $C$ tend to be pruned.} We, therefore, pick a pivot to maximize the number of saturated vertices in $P$. Specifically, we try to find the closest-to-saturation pivot $v_p$ in $P$ (Line~\ref{alg2:line9}), and then find a non-neighbor of $v_p$ in $C$ that is closer to saturation (Lines~\ref{alg2:line16}) as the new pivot $v_{new}\in C$, which is then used to expand $P$ (Line~\ref{alg2:line19}). While if the closest-to-saturation pivot $v_p$ cannot be found in $P$, we then pick $v_p$ in $C$ (Line~\ref{alg2:line10}), which is then used to expand $P$ (Line~\ref{alg2:line19}). 
Example~\ref{example1} illustrates the process of our pivot selection strategy.
\begin{example}[Pivot Selection]\label{example1}
Consider the graph $G$ shown in Figure~\ref{fig::example_graph} and $k=2$. Assume that the current $k$-plex $P = \{v_1, v_3\}$ and the candidate set $C = \{v_2, v_5, v_7\}$. Then according to Lines~\ref{alg2:line7}--\ref{alg2:line8}, $M_0 = \{ v_3 \}$ and $M = \{ v_3 \}$. Note that $M \cap P = \{ v_3 \} \neq \emptyset$, so $v_p = v_3$. Following our re-picking strategy, the new pivot vertex is selected from $\overline{N_C}(v_3) = \{ v_5, v_7 \}$, and the selected pivot vertex is $v_7$. 
\end{example}

\section{Upper Bounding and Pruning}\label{sec:prune}
This section introduces our upper bounding and additional pruning techniques used in Algorithms~\ref{alg::1} and~\ref{alg::2}, respectively, that are critical in speeding up the enumeration process. 
Due to space limitation, we put most of the proofs in the appendix of our full version~\cite{fullpaper}.

\vspace{1mm}
\noindent{\bf Seed Subgraph Pruning.} The theorem below states the second-order property of two vertices in a $k$-plex with size constraint:
\begin{theorem}
\label{lemma::2nd_order}
Let $P$ be a $k$-plex with $|P|\ge q$. Then, for any two vertices $u, v \in P$, we have (i)~if $(u,v)\not\in E$, $|N_P(u)\cap N_P(v)|\ge q-2k+2$, (ii)~otherwise, $|N_P(u)\cap N_P(v)|\ge q-2k$. 
\end{theorem}
\begin{proof}
    Please see Appendix~\ref{app:th1}~\cite{fullpaper}.
\end{proof}

Note that by setting $q=2k-1$, Case~(i) gives $|N_P(u)\cap N_P(v)|\ge (2k-1)-2k+2=1$, i.e., for any two vertices $u,v\in P$ that are not mutual neighbors, they must share a neighbor and is thus within 2 hops, which proves Theorem~\ref{lemma::diameter}.

This also gives the following corollary to help further prune the size of a seed subgraph $G_i$ in Line~\ref{algo1:line6} of Algorithm~\ref{alg::1}, which is also essential for our time complexity analysis (c.f., Lemma~\ref{th:subtasks}).
\begin{corollary}
    \label{corollary:2nd_order}
    Consider an undirected graph $G=(V,E)$ and an ordering of $V$: $\{v_1,v_2,..,v_n\}$. Let $v_i$ be the seed vertex and $G_i$ be the seed subgraph. Then, any vertex $u\in V_i$ (recall Algorithm~\ref{alg::1} Line~\ref{algo1:line4} for the definition of $V_i$) that satisfies either of the following two conditions can be pruned: 
    \begin{itemize}
    \item $u\in N_{G_i}(v_i)$ and $|N_{G_i}(u)\cap N_{G_i}(v_i)|<q-2k$;
    \item $u\in N^2_{G_i}(v_i)$ and $|N_{G_i}(u)\cap N_{G_i}(v_i)|<q-2k+2$.
    \end{itemize}
\end{corollary}

\vspace{1mm}
\noindent{\bf Upper Bound Computation.} We next consider how to obtain the upper bound of the maximum size of a $k$-plex containing $P$, which is called in Algorithm~\ref{alg::2} Line~\ref{alg2:line17}.

\begin{theorem}
\label{lemma::bound2}
Given a $k$-plex $P$ in a seed subgraph $G_i=(V_i, E_i)$, the upper bound of the maximum size of a $k$-plex containing $P$ is $\min_{u\in P} \{d_{G_i}(u)\} + k$.
\end{theorem}
\setlength{\textfloatsep}{5pt}
\begin{proof}
We illustrate the proof process using Figure~\ref{new} (where irrelevant edges are omitted). Let $P_{m}\subseteq P\cup C$ be a maximum $k$-plex containing $P$. Given any $u\in P$, we can partition $P_m$ into two sets: (1)~$N_{G_i}(u)\cap P_m$ (see the red vertices in Figure~\ref{new}), and (2)~$\overline{N_{G_i}}(u)\cap P_m$ (see the green vertices in Figure~\ref{new} including $u$ itself). The first set $N_{G_i}(u)\cap P_m$ has size at most $|N_{G_i}(u)|=d_{G_i}(u)$ (i.e., at most all the 7 neighbors in $N_{G_i}(u)$ in Figure~\ref{new} are included into $P_{m}$). For the second set, if more than $k$ vertices are included into $P_m$, then $u\in P_m$ is a non-neighbor of $k$ vertices in $P_m$, so $P_m$ violates the $k$-plex definition (Definition~\ref{definition}) which leads to a contradiction; as a result, at most $k$ vertices in $\overline{N_{G_i}}(u)\cap P_m$ (including $u$ itself already in $P_m$) can be added to $P_m$. Putting things together, $|P_m|=|N_{G_i}(u)\cap P_m|+|\overline{N_{G_i}}(u)\cap P_m|\leq d_{G_i}(u) + k$.  
Since $u$ can be an arbitrary vertex in $P$, we have $|P_m|\leq \min_{u\in P} \{d_{G_i}(u)\} + k $.
\end{proof}

We use our running example to illustrate how to use Theorem~\ref{lemma::bound2}.
\begin{example}\label{example2}
Consider the graph $G$ shown in Figure~\ref{fig::example_graph} and $k=2$, and assume that $P = \{ v_1, v_3 \}$. For $v_1$, we can add at most all its 3 neighbors $\{v_2, v_5, v_7\}$ and $k=2$ non-neighbors into a $k$-plex containing $P$. Thus, the upper bound of its size is $3+2=5$. Similarly, for $v_3$, the upper bound is $2+2=4$. Therefore, the size of the $k$-plex expanded from $P$ is at most $\min\{5, 4\} = 4$. 
\end{example}

\begin{figure}[t]
\centering
\includegraphics[width=0.75\columnwidth]{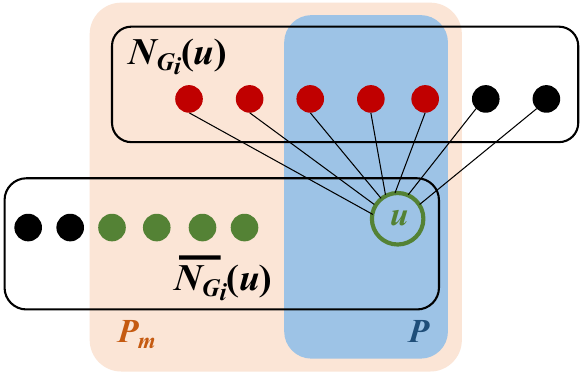}
\caption{Upper Bound Illustration for Theorem~\ref{lemma::bound2}}\label{new}
\end{figure}

In Line~\ref{alg2:line17} of Algorithm~\ref{alg::2}, we compute $\min_{u\in P\cup\{v_p\}} \{d_{G_i}(u)\} + k$ as an upper bound by Theorem~\ref{lemma::bound2}. 
Recall that we already select $v_p$ as the vertex in $G_i$ with the minimum degree in Line~\ref{alg2:line7} of Algorithm~\ref{alg::2}, so the upper bound can be simplified as $d_{G_i}(v_p)+k$. Note that $v_p$ here is the one obtained in Lines~\ref{alg2:line7}--\ref{alg2:line10}, not the $v_{new}$ that replaces the old $v_p$ in Line~\ref{alg2:line16} in the case when $v_p\in P$. 

We next derive another upper bound of the maximum size of a $k$-plex containing $P$. First, we define the concept of ``support number of non-neighbors''. Given a $k$-plex $P$ and candidate set $C$, for a vertex $v\in P\cup C$, its support number of non-neighbors is defined as $\text{sup}_P(v)= k-\overline{d_P}(v)$, which is the maximum number of non-neighbors of $v$ {\bf outside} $P$ that can be included in any $k$-plex containing $P$.

\begin{theorem}
\label{lemma::bound4}
Let $v_i$ and $G_i=(V_i, E_i)$ be the seed vertex and corresponding seed subgraph, respectively, and consider a sub-task $P_S=S\cup\{v_i\}$ where $S\subseteq N^2_{G_i}(v_i)$.

For a $k$-plex $P$ satisfying $P_S\subseteq P\subseteq V_i$ and for a pivot vertex $v_p\in C\subseteq C_S=N_{G_i}(v_i)$, the upper bound of the maximum size of a $k$-plex containing $P\cup \{v_p\}$ is
    \begin{equation}\label{eq:bound4}
        |P|+\text{sup}_P(v_p)+|K|,
    \end{equation}
where the set $K$ is computed as follows:

Initially, $K=N_{C}(v_p)$. For each $w\in K$, we find $u_m\in \overline{N_P}(w)$ such that $\text{sup}_P(u_m)$ is the minimum; if $\text{sup}_P(u_m)>0$, we decrease it by 1. Otherwise, we remove $w$ from $K$.
\end{theorem}
\setlength{\textfloatsep}{5pt}

Figure~\ref{ub2} shows the rationale of the upper bound in Eq~(\ref{eq:bound4}), where $v_p\in C_S$ is to be added to $P$ (shown inside the dashed contour). Let $P_{m}\subseteq P\cup C$ be a maximum $k$-plex containing $P\cup \{v_p\}$, then the three terms in Eq~(\ref{eq:bound4}) correspond to the upper bounds of the three sets that $P_m$ can take its vertices from: (1)~$P$ whose size is exactly $|P|$, (2) $v_p$'s non-neighbors in $C$, i.e., $P_{m}\cap\overline{N_C}(v_p)$ (including $v_p$ itself) whose size is upper-bounded by $\text{sup}_P(v_p)= k-\overline{d_P}(v_p)$, and (3)~$v_p$'s neighbors in $C$, i.e., $P_{m}\cap N_C(v_p)$ whose size is at most $|K|$.

Here, $K$ computes the largest set of candidates in $N_C(v_p)$ that can expand $P\cup \{v_p\}$. Specifically, for each $w\in N_C(v_p)$, if there exists a non-neighbor in $P$,  denoted by $u_m$, that has $\text{sup}_P(u_m)=0$, then $w$ is pruned (from $K$) since if we move $w$ to $P$, $u_m$ would violate the $k$-plex definition. Otherwise, we decrement $\text{sup}_P(u_m)$ to reflect that $w$ (which is a non-neighbor of $u_m$) has been added to $K$ (i.e., removed from $C$) to expand $P\cup \{v_p\}$. Algorithm~\ref{alg::3} shows the above approach to compute the upper bound, which is called in Line~\ref{alg2:line17} of Algorithm~\ref{alg::2}.


\begin{figure}[t]
\centering
\includegraphics[width=0.65\columnwidth]{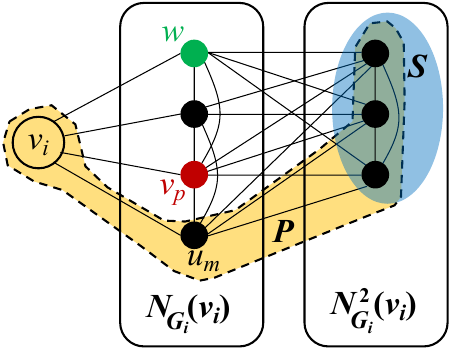}
\caption{Upper Bound Illustration for Theorem~\ref{lemma::bound4}}\label{ub2}
\end{figure}

\begin{algorithm}
    \DontPrintSemicolon
    \SetKwBlock{Begin}{function}{end function}
        $\text{sup}_P(v_p)\leftarrow k-\overline{d_{P}}(v_p)$\label{alg3:line1}\;
        \lForEach{$u\in P$}{\ \ \ \ $\text{sup}_P(u)\leftarrow k-\overline{d_P}(u)$}\label{alg3:line2}
        $ub\leftarrow |P|+\text{sup}_P(v_p)$\label{alg3:line3}\;
        \ForEach{$w\in N_{C}(v_p)$}{\label{alg3:line4}
            Find $u_m\in \overline{N_P}(w)$ s.t.\ $\text{sup}_P(u_m)$ is the minimum\label{alg3:line5}\;
            \If{$\text{sup}_P(u_m)>0$}{\label{alg3:line6}
                $\text{sup}_P(u_m)\leftarrow \text{sup}_P(u_m)-1$\label{alg3:line7}\;
                $ub\leftarrow ub+1$\label{alg3:line8}
            }
        }
        \Return $ub$\;
\caption{Computing Upper Bound by Theorem~\ref{lemma::bound4}}    
\label{alg::3}
\end{algorithm}
\setlength{\textfloatsep}{5pt}


We next prove that $K$ is truly the largest set of candidates in $N_C(v_p)$ that can expand $P\cup \{v_p\}$. 
\begin{proof}
Please see Appendix~\ref{app:th3}~\cite{fullpaper}.
\end{proof}
Putting Theorems~\ref{lemma::bound2} and~\ref{lemma::bound4} together, the upper bound in Line~\ref{alg2:line17} of Algorithm~\ref{alg::2} is given by
\begin{equation}\label{eq:our_ub}
\min\{|P|+\text{sup}_P(v_p)+|K|, d_{G_i}(v_p)+k\}.
\end{equation}

We use our running example to illustrate how to use Theorem~\ref{lemma::bound4}.
\begin{example}
    We use the same graph and settings as above two examples, i.e., $k=2$, $P= \{v_1, v_3 \}$, and $C = \{v_2, v_5, v_7 \}$. According to Example~\ref{example1}, the pivot vertex is $v_7$. We can calculate $\sup_P(v_7) = k - \overline{d_P}(v_7) = 1$ since $v_7$ only has one non-neighbor $v_3$ in $P$ (i.e., $\overline{d_P}(v_7) = 1$). Also, $K$ is initialized as $N_C(v_7) = \{v_5\}$. For $w=v_5$, since $\overline{N_P}(v_5) = \{v_3\}$, we have $u_m=v_3$. Since $\sup_P(v_3) = k - \overline{d_P}(v_3) = 0$ (as $v_3$ has two non-neighbors $\{v_1, v_3\}$ in $P$), $w=v_5$ is removed from $K$, so $K=\emptyset$. Thus, the upper bound of the size of the $k$-plex expanded from $P$ is $|P| + \sup_P(v_7) + |K|= 2 + 1 + 0 = 3$. 
\end{example}

For an initial sub-task $\langle P_S, C_S, X_S\rangle$, we can further improve Theorem~\ref{lemma::bound4} as follows.

\begin{theorem}


Let $v_i$ and $G_i=(V_i, E_i)$ be the seed vertex and corresponding seed subgraph, respectively, and consider a sub-task $P_S=S\cup\{v_i\}$ where $S\subseteq N^2_{G_i}(v_i)$ and $S\neq\emptyset$.
We caculate $|K|$ with a modified version of Algorithm~\ref{alg::3} with $v_p=v_i$, and $\sup_P(v_i)=0$ in Line~1. The upper bound is $|P_S|+\sup_P(v_i)+|K|=|P_S|+|K|$.
\label{reduction4}
\end{theorem}
\begin{proof}
Please see Appendix~\ref{app:th4}~\cite{fullpaper}.
\end{proof}

Intuitively, this is a special case of Theorem~\ref{lemma::bound4} with $P = S$ and $v_p = v_i$, and $\sup_P(v_i)=0$ since there does not exist any non-neighbor of $v_i$ in $C$ to be added to $P$ (recall from Line~\ref{algo1:line8} of Algorithm~\ref{alg::1} that $C$ only contains $v_i$'s neighbors).

Recall that Theorem~\ref{lemma::bound2} also gives another upper bound of the maximum size of a $k$-plex containing $P_S$, which is $\min_{v\in P_S} \{d_{G_i}(v)\}+k$. 
Combining with Theorem~\ref{reduction4}, the new upper bound is given by $ub(P_S) = \min\left\{|P_S|+|K|, \min_{v\in P_S} \{d_{G_i}(v)\}+k\right\}$.

Right before Line~\ref{algo1:line10} of Algorithm~\ref{alg::1}, we check if $ub(P_S)<q$; if so, we prune this sub-task without calling Branch(.).

\vspace{1mm}
\noindent {\bf Time Complexity Analysis.} We now analyze the time complexity of our algorithm, i.e., Algorithm~\ref{alg::1}. Recall that $D$ is the degeneracy of $G$, that $\Delta$ is the maximum degree of $G$, and that seed vertices are in the degeneracy ordering of $V$, $\eta=\{v_1,\dots, v_{n}\}$. Therefore, given a seed subgraph $G_i=(V_i, E_i)$ where $V_i\subseteq\{v_i,v_{i+1},\dots,v_{n}\}$, for each $v\in G_i$, we have $d_{G_i}(v)=|N_{G_i}(v)|\leq D$.

Let us first consider the time complexity of Algorithm~\ref{alg::3}.
\begin{lemma}\label{th:complexity_alg3}
The time complexity of Algorithm~\ref{alg::3} is given by $$O(k+(k+1)D)\approx O(D).$$
\end{lemma}
\begin{proof}
Please see Appendix~\ref{app:th5}~\cite{fullpaper}.
\end{proof}

We next bound the number of sub-tasks in each $G_i$ created by Lines~\ref{algo1:line7} and~\ref{algo1:line10} of Algorithm~\ref{alg::1}.

\begin{lemma}\label{th:subtasks}
Let $v_i$ and $G_i=(V_i, E_i)$ be the seed vertex and corresponding seed subgraph, respectively. Also, let us abuse the notation $|N_{G_i}^2(v_i)|$ to mean the one pruned by Corollary~\ref{corollary:2nd_order} in Line~\ref{algo1:line6} of Algorithm~\ref{alg::1}. Then, we have $|N_{G_i}(v_i)|\le D$ and $|N_{G_i}^2(v_i)|=O\left(r_1\right)$ where $r_1=\min\left\{\frac{D\Delta}{q-2k+2},n\right\}$.

Also, the number of subsets $S\subseteq N_{G_i}^2(v_i)$ ($|S|\leq k-1$) is bounded by $O\left(|N_{G_i}^2(v_i)|^k\right)=O\left(r_1^k\right)$. 
\end{lemma}
\begin{proof}
Please see Appendix~\ref{app:th6}~\cite{fullpaper}.
\end{proof}

We can also bound the time complexity of Algorithm~\ref{alg::2}:
\begin{lemma}\label{th:times}
Let $v_i$ and $G_i=(V_i, E_i)$ be the seed vertex and corresponding seed subgraph, respectively. Then, Branch$(G_i,k,q,P_S,C_S,$ $X_S)$ (see Line~\ref{algo1:line10} in Algorithm~\ref{alg::1}) recursively calls the body of Algorithm~\ref{alg::3} for $O(\gamma_k^{D})$ times, where $\gamma_k<2$ is the maximum positive real root of
$x^{k+2}-2x^{k+1}+1=0$ (e.g., $\gamma_1= 1.618$, $\gamma_2 = 1.839$, and $\gamma_3 = 1.928$).
\end{lemma}

To see this bound, note that Theorem~1 of~\cite{cikm22maximal} has proved that the branch-and-bound procedure is called for $O(\gamma_k^{|C|})$ times. In~\cite{cikm22maximal}, the candidate set $C$ is taken from vertices within two hops away from each seed vertex $v_i$, so the branch-and-bound procedure is called for $O(\gamma_k^{|C|})\leq O(\gamma_k^{n})$ times. In our case, $C=C_S=N_{G_i}(v_i)$ which is much tighter since $|C|\leq D$, hence the branch-and-bound procedure is called for $O(\gamma_k^{|C|})\leq O(\gamma_k^{D})$ times.

Finally, consider the cost of the recursion body of Algorithm~\ref{alg::2}. Note that besides $d_{P}(.)$, we also maintain $d_{G_i}(.)$ for all vertices in $G_i$, so that Line~\ref{alg2:line7} of Algorithm~\ref{alg::2} (the same applies to Line~\ref{alg2:line16}) can obtain the vertices with minimum $d_{G_i}(.)$ in $O(|P|+|C|)\approx O(D)$ time. This is because $P\cup C\subseteq P_S\cup C_S$, so $O(|P|+|C|)=O(|P_S|+|C_S|)=O(k+D)\approx O(D)$, as $|P_S|\leq k$ and $|C_S|\leq D$.

As for the tightening of $C$ and $X$ in Lines~\ref{alg2:line2}--\ref{alg2:line3} of Algorithm~\ref{alg::2} (the same applies to Line~\ref{alg2:line12}), the time complexity is $O(|P|(|C|+|X|))$. Specifically, we first compute the set of saturated vertices in $P$, denoted by $P^*$. Since we maintain $d_P(.)$, we can find $P^*$ in $O(|P|)$ time by examining if each vertex $u$ has $d_P(u)=|P|-k$. Then, for each vertex $v\in C\cup X$, we do not prune it iff (1)~$v$ is adjacent to all vertices in $P^*$, and meanwhile, (2)~$d_P(v)\geq |P\cup\{v\}|-k=|P|+1-k$. This takes $O(|P^*|(|C|+|X|))=O(|P|(|C|+|X|))$ time.

The recursive body takes time $O(|P|(|C|+|X|))$ which is dominated by the above operation. Note that by Lemma~\ref{th:complexity_alg3}, Algorithm~\ref{alg::2} Line~\ref{alg2:line17} takes only $O(D)\approx O(|P|+|C|)$ time, and the time to select pivot (cost dominated by Line~\ref{alg2:line7}) also takes only $O(|P|+|C|)$ time.

Now we are ready to present the time complexity of Algorithm~\ref{alg::1}.
\begin{theorem}\label{th:complexity}
Given an undirected graph $G=(V,E)$ with degeneracy $D$ and maximum degree $\Delta$, Algorithm~\ref{alg::1} lists all the $k$-plexes with size at least $q$ within time $O\left(nr_1^kr_2\gamma_k^{D}\right)$, where $r_1=\min\left\{\frac{D\Delta}{q-2k+2},n\right\}$ and $r_2 = \min\left\{\frac{D\Delta^2}{q-2k+2},nD\right\}$.
\end{theorem}

\begin{proof}
Please see Appendix~\ref{app:th7}~\cite{fullpaper}.
\end{proof}

\vspace{1mm}
\noindent {\bf Additional Pruning by Vertex Pairs.} We next present how to utilize the further property between vertex pairs in $G_i$ to enable three further pruning opportunities, all based on Lemma~\ref{lemma::bound1} below.

\begin{lemma}
\label{lemma::bound1}
    Given a $k$-plex $P$ and candidate set $C$, the upper bound of the maximum size of a $k$-plex containing $P$ is
    \begin{equation*}
        \underset{u,v\in P}{\text{min}}\Big\{|P|+\text{sup}_P(u)+\text{sup}_P(v)+|N_{u}(C)\cap N_{v}(C)|\Big\}
    \end{equation*}
\end{lemma}
\begin{proof} 
Please see Appendix~\ref{app:th8}~\cite{fullpaper}.
\end{proof}

Recall that Algorithm~\ref{alg::1} Line~\ref{algo1:line7} enumerates set $S$ from the vertices of $N_{G_i}^2(v_i)$. 
The first pruning rule below checks if two vertices $u_1$, $u_2\in N_{G_i}^2(v_i)$ have sufficient common neighbors in $C_S$, and if not, then $u_1$ and $u_2$ cannot occur together in $S$.

\begin{theorem}
    Let $v_i$ be a seed vertex and $G_i=(V_i, E_i)$ be the corresponding seed subgraph. For any two vertices $u_1$, $u_2\in N_{G_i}^2(v_i)$, if either of the following conditions are met
    \begin{itemize}
        \item $(u_1,u_2)\in E_i$ and $|N_{C_S}(u_1)\cap N_{C_S}(u_2)|<q-k-2 \cdot \max\{k-2,0\}$, 
        \item $(u_1,u_2)\notin E_i$ and $|N_{C_S}(u_1)\cap N_{C_S}(u_2)|<q-k-2 \cdot \max\{k-3,0\}$, 
    \end{itemize}
then $u_1$ and $u_2$ cannot co-occur in a $k$-plex $P$ with $|P|\geq q$. 
\label{reduction1}
\end{theorem}

\begin{proof} 
Please see Appendix~\ref{app:th9}~\cite{fullpaper}.
\end{proof}
   
Similar analysis can be adapted for the other two cases: (1)~$u_1\in N_{v_i}(G_i)$ and $u_2\in N^2_{v_i}(G_i)$, and (2)~$u_1$, $u_2\in N_{v_i}(G_i)$, which we present in the next two theorems.

\begin{theorem}
Let $v_i$ be a seed vertex and $G_i=(V_i, E_i)$ be the corresponding seed subgraph. For any two vertices $u_1\in N^2_{G_i}(v_i)$ and $u_2\in N_{G_i}(v_i)$, let us define $C_S^-=C_S-\{u_2\}$, then if either of the following two conditions are met 
    \begin{itemize}
        \item $(u_1,u_2)\in E_i$ and $|N_{C_S^-}(u_1)\cap N_{C_S^-}(u_2)|<q-2k-2\cdot\max\{k-2,0\}$,
        \item $(u_1,u_2)\notin E_i$ and $|N_{C_S^-}(u_1)\cap N_{C_S^-}(u_2)|<q-k-\max\{k-2,0\}-\max\{k-2,1\}$,
    \end{itemize}
then $u_1$ and $u_2$ cannot co-occur in a $k$-plex $P$ with $|P|\geq q$.
\label{reduction3}
\end{theorem}

\begin{proof}
Please see Appendix~\ref{app:th10}~\cite{fullpaper}.
\end{proof}

\begin{theorem}
Let $v_i$ be a seed vertex and $G_i=(V_i, E_i)$ be the corresponding seed subgraph. For any two vertices $u_1$, $u_2\in N_{G_i}(v_i)=C_S$, let us define $C_S^-=C_S-\{u_1, u_2\}$, then if either of the following two conditions are met 
    \begin{itemize}
        \item $(u_1,u_2)\in E_i$ and $|N_{C_S^-}(u_1)\cap N_{C_S^-}(u_2)|<q-3k$,
        \item $(u_1,u_2)\notin E_i$ and $|N_{C_S^-}(u_1)\cap N_{C_S^-}(u_2)|<q-k-2\cdot\max\{k-1,1\}$,
    \end{itemize}
then $u_1$ and $u_2$ cannot co-occur in a $k$-plex $P$ with $|P|\geq q$.
\label{reduction2}
\end{theorem}
\begin{proof}
Please see Appendix~\ref{app:th11}~\cite{fullpaper}.
\end{proof}

We next explain how Theorems~\ref{reduction1}, \ref{reduction3} and~\ref{reduction2} are used in our algorithm to prune the search space. Recall that $G_i$ is dense, so we use adjacency matrix to maintain the information of $G_i$. Here, we also maintain a boolean matrix $T$ so that for any $u_1,u_2\in V_i$, $T[u_1][u_2]=$ {\em false} if they are pruned by Theorem~\ref{reduction1} or~\ref{reduction3} or~\ref{reduction2} due to the number of common neighbors in the candidate set being below the required threshold; otherwise, $T[u_1][u_2]=$ {\em true}. Note that given $T$, we can obtain $T[u_1][u_2]$ in $O(1)$ time to determine if $u_1$ and $u_2$ can co-occur.

Recall from Figure~\ref{set_enum} that we enumerate $S$ via a set-enumeration tree. When we enumerate $S$ in Algorithm~\ref{alg::1} Line~\ref{algo1:line7}, assume that the current $S$ is expanded from $S'$ by adding $u$, and let $ext(S')$ be those candidate vertices that can still expand $S'$, then by Theorem~\ref{reduction1}, we can incrementally prune those candidate vertices $u'\in ext(S')$ with $T[u][u']=$ {\em false} to obtain $ext(S)$ that can expand $S$ further.

We also utilize Theorem~\ref{reduction3} to further shrink $C_S$ in Algorithm~\ref{alg::1} Line~\ref{algo1:line8}. Assume that the current $S$ is expanded from $S'$ by adding $u$, then we can incrementally prune those candidate vertices $u'\in C_{S'}$ with $T[u][u']=$ {\em false} to obtain $C_S$.

Finally, we utilize Theorem~\ref{reduction2} to further shrink $C$ and $X$ in Algorithm~\ref{alg::2} Lines~\ref{alg2:line2} and~\ref{alg2:line3}. Specifically, assume that $v_p$ is newly added to $P$, then Line~\ref{alg2:line2} now becomes
$$C\leftarrow\{v\in C\ \wedge\ P\cup\{v\} \text{ is a }k\text{-plex}\ \wedge\ T[v_p][v]= \mbox{\em true}\}.$$
Recall from Algorithm~\ref{alg::1} Line~\ref{algo1:line9} that vertices in $X$ may come from $V'_i$ or $V_i$. So for each $v\in X$, if $v\in V_i$, then we prune $v$ if $T[v_p][v]= \mbox{\em false}$. This is also applied in Line~\ref{alg2:line12} when we compute the new exclusive set to check maximality.

\vspace{1mm}
\noindent {\bf Variant of the Proposed Algorithm.} Recall from Algorithm~\ref{alg::2} Lines~\ref{alg2:line15}-\ref{alg2:line16} that we always select the pivot $v_p$ to be from $C$, so that our upper-bound-based pruning in Lines~\ref{alg2:line17}-\ref{alg2:line18} can be applied to try to prune the branch in Line~\ref{alg2:line19}. In fact, if $v_p\in P$, FaPlexen~\cite{aaai2020maximal} proposed another branching method to reduce the search space, also adopted by ListPlex~\cite{www22maximal}. Specifically, let us define $s=\text{sup}_P(v_p)$, and $\overline{N_{C}}(v_p)=\{w_1, w_2, \ldots, w_\ell\}$, then we can move at most $s$ vertices from $\overline{N_{C}}(v_p)$ to $P$ to produce $k$-plexes. Note that $s\leq k$ since $\text{sup}_P(v_p)= k-\overline{d_P}(v_p)$ so $s$ is small, and that $s\leq\ell$ since otherwise, $P\cup C$ is a $k$-plex (because $s=\text{sup}_P(v_p)>\ell=|\overline{N_{C}}(v_p)|=\overline{d_{C}}(v_p)$ means $k-(|P|-d_P(v_p))>|C|-d_C(v_p)$, i.e., $d_{P\cup C}>|P|+|C|-k$) and this branch of search terminates (see Algorithm~\ref{alg::2} Lines~\ref{alg2:line11}-\ref{alg2:line14}). Therefore, let the current task be $\langle P,C,X\rangle$, then it only needs to produce $s+1$ branches without missing $k$-plexes:
\begin{equation}\label{eq:branch1}
\langle P,\ C-\{w_1\},\ X\cup\{w_1\}\rangle,
\end{equation}
For $i=2, \cdots, s$,
\begin{equation}\label{eq:branch2}
\langle P\cup\{w_1,\ldots,w_{i-1}\},\ C-\{w_1,\ldots,w_i\},\ X\cup\{w_i\}\rangle,
\end{equation}
\begin{equation}\label{eq:branch3}
\langle P\cup\{w_1,\ldots,w_\ell\},\ C-\{w_1,\ldots,w_\ell\},\ X\rangle.
\end{equation}

In summary, if $v_p\in P$, we can apply Eq~(\ref{eq:branch1})--Eq~(\ref{eq:branch3}) for branching, while if $v_p\in C$, we can apply the upper bound defined by Eq~(\ref{eq:our_ub}) to allow the pruning of $\langle P\cup\{v_p\}, C-\{v_p\}, X\rangle$ in Algorithm~\ref{alg::2} Line~\ref{alg2:line19}.

Therefore, besides the original Algorithm~\ref{alg::2} denoted by  {\bf Ours}, we also consider a variant of Algorithm~\ref{alg::2} which, when $v_p\in P$ is selected by Lines~\ref{alg2:line7}--\ref{alg2:line10}, uses Eq~(\ref{eq:branch1})--Eq~(\ref{eq:branch3}) for branching rather than re-picking a pivot $v_p\in C$ as in Lines~\ref{alg2:line15}--\ref{alg2:line16}. We denote this variant by {\bf Ours\_P}. As we shall see in Section~\ref{sec:results}, Ours\_P is generally not as time-efficient as Ours, showing that upper-bound based pruning is more effective than the branch reduction scheme of Eq~(\ref{eq:branch1})--Eq~(\ref{eq:branch3}), so Ours is selected as our default algorithm in our experiments.

\section{Parallelization}\label{sec:parallel}

Recall from Algorithm~\ref{alg::1} that we generate initial task groups $T_{\{v_i\}}$ each creating and maintaining $G_i$. The sub-tasks of $T_{\{v_i\}}$ are $T_{\{v_i\}\cup S}$ that are generated by enumerating $S\subseteq N_{G_i}^2(v_i)$, and each such task runs the recursive $Branch(.)$ procedure of Algorithm~\ref{alg::2} (recall Line~\ref{algo1:line10} of Algorithm~\ref{alg::1}).

\begin{figure}[t]
\centering
\includegraphics[width=0.8\columnwidth]{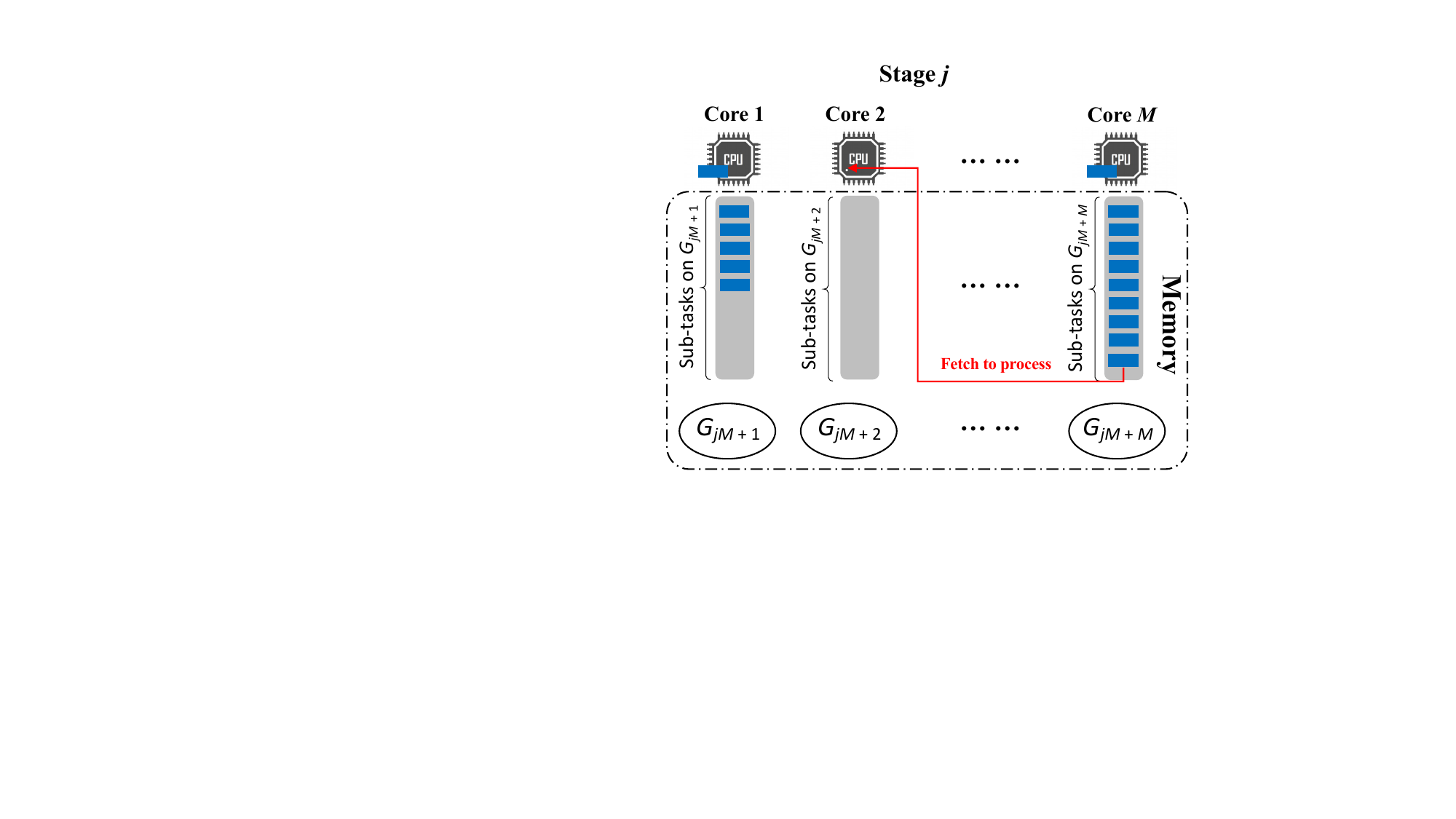}
\caption{Illustration of Parallel Processing}\label{parallel}
\end{figure}

We parallelize Algorithm~\ref{alg::1} on a multi-core machine with $M$ CPU cores (and hence $M$ working threads to process tasks) in stages. In each stage~$j$ ($j=0, 1, \cdots$), the $M$ working threads obtain $M$ new task groups generated by the next $M$ seed vertices in $\eta$ for parallel processing, as illustrated in Figure~\ref{parallel}.

Specifically, at the beginning of Stage~$j$, the $i\textsuperscript{th}$ thread creates and processes the task group with seed vertex $v_{jM+i}$ by creating the subgraph $G_{jM+i}$, and enumerating $S$ to create the sub-tasks with $P_S=\{v_{jM+i}\}\cup S$ and adding them into a task queue $Q_i$ local to Thread~$i$. Then, each thread~$i$ processes the tasks in its local queue $Q_i$ to maximize data locality and hence cache hit rate (since the processing is on a shared graph $G_i$). If Thread~$i$ finishes all tasks in $Q_i$ but other threads are still processing their tasks in Stage~$j$, then Thread~$i$ will obtain tasks from another non-empty queue $Q_{i'}$ for processing to take over some works of Thread~$i'$. This approach achieves load balancing while maximizing the CPU cache hit rate.

Stage~$j$ ends when the tasks in all queues are exhausted, after which we release the memory used by these task groups (e.g., seed subgraphs $G_i$) and move forward to Stage~$(j+1)$ to process the next $M$ seed vertices. The stages are repeated until all seed vertices in $\eta$ are exhausted.

So far, we treat each sub-task $T_{P_S}$ as an independent task run by a thread in its entirety. However, some sub-tasks $T_{P_S}$ can become stragglers that take much longer time to complete than other tasks (e.g., due to a much larger set-enumeration subtree under $P_S$). We propose to use a timeout mechanism to further decompose each straggler task into many smaller tasks to allow parallel processing. Specifically, let $t_0$ be the time when the current task is created, and let $t_{cur}$ be the current time. Then in Algorithm~\ref{alg::2} Line~\ref{alg2:line19} (resp.\ Line~\ref{alg2:line20}), we only recursively call $Branch(.)$ over $\langle P\cup\{v_p\}, C-\{v_p\}, X\rangle$ (resp.\ $\langle P, C-\{v_p\}, X\cup\{v_p\}\rangle$) if $t_{cur}-t_0\leq\tau_{time}$, where $\tau_{time}$ is a user-defined task timeout threshold. Otherwise, let the thread processing the current task be Thread~$i$, then we create a new task and add it to $Q_i$. The new tasks can reuse the seed subgraph of its task group, but need to materialize new status variables such as containers for keeping $P$, $C$ and $X$, and the boolean matrix $T[u_1][u_2]$ for pruning by vertex pairs. 

In this way, a straggler task will call $Branch(.)$ for recursive backtracking search as usual until $t_{cur}-t_0\leq\tau_{time}$, after which it backtracks and creates new tasks to be added to $Q_i$. If a new task also times out, it will be further decomposed in a similar manner, so stragglers are effectively eliminated at the small cost of status variable materialization.
 
\section{Experiment}\label{sec:results}
In this section, we conduct comprehensive experiments to evaluate our method for large maximal $k$-plex enumeration, and compare it with the other existing methods. We also conduct an ablation study to show the effectiveness of our optimization techniques. 

\vspace{1mm}

\noindent \textbf{Datasets and Experiment Setting. }
Following~\cite{d2k,www22maximal,cikm22maximal}, we use 18 real-world datasets in our experiments as summarized in Table~\ref{table::dataset}, where $n$ and $m$ are the numbers of vertices and edges, respectively; $\Delta$ indicates the maximum degree and $D$ is the degeneracy. These public graph datasets are obtained from Stanford Large Network Dataset Collection (SNAP)~\cite{snapnets} and the Laboratory for Web Algorithmics (LAW)~\cite{BoVWFI,law}. Similar to the previous works~\cite{d2k,www22maximal,cikm22maximal}, we roughly categorize these graphs into three types: small, medium, and large. The ranges of the number of vertices for these three types of graphs are $[1, 10^4)$, $[10^4, 5 \times 10^6)$, and $[5 \times 10^6, +\infty)$. 

Our code is written in C++14 and compiled by g++-7.2.0 with the -O3 flag. All the experiments are conducted on a platform with 24 cores (Intel Xeon Gold 6248R CPU 3.00GHz) and 128GB RAM.

\renewcommand{\arraystretch}{1.1} 
\begin{table}[t]
\centering
  \caption{Datasets}
  \label{tbl-data}
  \resizebox{0.84\columnwidth}{!}{
      \begin{tabular}{c|c|c|c|c}
        \toprule[2pt]
        Network & $n$ & $m$ & $\Delta$ & D \\
        \hline
        jazz & 198 & 2742 & 100 & 29 \\
        wiki-vote & 7115 & 100,762 & 1065 & 53\\
        lastfm & 7624 & 27,806 & 216 & 20\\
        \hline
        as-caida & 26,475 & 53,381 & 2628 & 22\\
        soc-epinions & 75,879 & 405,740 & 3044 & 67\\
        soc-slashdot & 82,168 & 504,230 & 2552 & 55\\
        email-euall & 265,009 & 364,481 & 7636 & 37\\
        com-dblp    & 317,080 & 1,049,866 & 343   & 113\\
        amazon0505 & 410,236 & 2,439,437 & 2760 & 10\\
        soc-pokec & 1,632,803 & 22,301,964 & 14,854 & 47\\
        as-skitter & 1,696,415 & 11,095,298 & 35,455 & 111\\
        \hline
        enwiki-2021 & 6,253,897 & 136,494,843 &232,410 & 178\\
        arabic-2005 & 22,743,881 & 553,903,073 & 575,628 & 3247\\
        uk-2005 & 39,454,463 & 783,027,125 & 1,776,858 & 588\\
        it-2004 & 41,290,648 & 1,027,474,947 & 1,326,744 & 3224\\
        webbase-2001 & 115,554,441 & 854,809,761 & 816,127 & 1506\\
		\bottomrule[2pt]
\end{tabular}
}
\label{table::dataset}
\end{table}
 \setlength{\textfloatsep}{10.5pt}

\renewcommand{\arraystretch}{1.1} 

\begin{table*}[t]
    \centering
    \caption{Running Time (sec) of Listing Large Maximal $k$-Plexes from Small and Medium Graphs by Various Algorithms}
    \resizebox{2.1\columnwidth}{!}{
    \setlength{\tabcolsep}{2mm}{
    \begin{tabular}{c|c|c|c|c|c|c|c|c|c|c|c|c|c|c|c}
	\toprule[2pt]
	\multirow{2}{*}{\tabincell{c}{ Network\\ $(n,m)$}} & \multirow{2}{*}{$k$}    & \multirow{2}{*}{$q$} & \multirow{2}{*}{\#$k$-plexes}  &  \multicolumn{4}{c|}{Running time (sec)}   & \multirow{2}{*}{\tabincell{c}{Network\\ $(n,m)$}} & \multirow{2}{*}{$k$}    & \multirow{2}{*}{$q$} & \multirow{2}{*}{\#$k$-plexes}  &  \multicolumn{4}{c}{Running time (sec)}\\
	\cline{5-8} \cline{13-16}
	& & &  & FP & ListPlex & Ours\_P & Ours & & & &  & FP & ListPlex & Ours\_P & Ours\\
	\hline
    \tabincell{c}{jazz (198, 2742)} & 4 & 12 & 2,745,953  & 3.68 & 4.12 & 3.92 & \bf 2.87 &\multirow{5}{*}{\tabincell{c}{soc-epinions\\ (75,879, 405,740)}} & \multirow{2}{*}{2} & 12 & 49,823,056  & 278.56  & 153.64 & 157.98 &\bf 130.14 \\
	\cline{1-8}
    \tabincell{c}{lastfm (7624, 27,806)} & 4 & 12 & 1,827,337  & 2.39 & 2.58 & 2.52 & \bf 2.04 &  &  & 20 & 3,322,167  & 16.65 & 17.00 &16.65 & \bf 14.01 \\
    \cline{1-8} \cline{10-16}
    \multirow{3}{*}{\tabincell{c}{as-caida\\(26,475, 53,381)}} & 2 & 12 & 5336  & \bf0.03 & \bf0.03 & \bf0.03 & \bf0.03 & &\multirow{2}{*}{3}  & 20 & 548,634,119  & 2240.68 & 2837.49 & 2442.10 &\bf1540.87 \\
    \cline{2-8}  
     & 3 & 12 & 281,251 & 0.94 & 0.78 & 0.67 &\bf0.53 & &  & 30 & 16,066  & 3.29 & 4.66 & 2.52 &\bf2.11\\
    \cline{2-8} \cline{10-16}
    & 4 & 12 & 15,939,891  & 51.39 & 45.31 &37.52 &\bf26.08 & & 4 & 30 & 13,172,906 &139.59 &545.82 &198.88 &\bf93.47\\
    \cline{1-16}
    \multirow{3}{*}{\tabincell{c}{amazon0505\\(410,236, 2,439,437)}} & 2 & 12 & 376  & 0.37 & 0.11 &0.08 &\bf0.07 &\multirow{6}{*}{\tabincell{c}{wiki-vote\\(7115, 100,762)}} &\multirow{2}{*}{2} & 12 & 2,919,931  & 19.21 & 14.77 &13.62 &\bf12.58\\
    \cline{2-8}
    & 3 & 12 & 6347  & 0.57 & 0.26 &\bf0.20 &0.22 & &  & 20 & 52 & \bf0.37 & 0.46 &0.44 &0.42 \\
    \cline{2-8} \cline{10-16}
    & 4 & 12 & 105,649  & 1.47 & 0.99 &0.91 &\bf0.78 & & \multirow{2}{*}{3} & 12 & 458,153,397  & 2680.68 & 2037.51 &1746.63 &\bf1239.83 \\
    \cline{1-8}   
    \multirow{4}{*}{\tabincell{c}{as-skitter\\(1,696,415, 11,095,298)}} & \multirow{2}{*}{2} & 60 & 87,767  & \bf4.37 & 5.14 &4.41 & 4.68 & & & 20 & 156,727  & 11.91 & 8.13 &4.67 & \bf4.15 \\
    \cline{10-16}
     & & 100 & 0  & 1.49 & 0.14 & 0.14 & \bf0.11 & &\multirow{2}{*}{4} & 20 & 46,729,532  & 483.97 & 1025.54 &455.37 &\bf252.40 \\
     \cline{2-8}
    &\multirow{2}{*}{3} & 60 & 9,898,234  & 283.52 & 1010.48 & 302.53 & \bf234.17 & & & 30 & 0  & \bf0.03 & 0.07 & 0.08 &0.06 \\
    \cline{9-16}
     &  & 100 & 0  & 1.49 & 0.14 &0.14 & \bf0.11 &\multirow{6}{*}{\tabincell{c}{soc-slashdot\\(82,168, 504,230)}} &\multirow{2}{*}{2} & 12 & 30,281,571  & 91.10 & 65.57 &66.16 &\bf51.41 \\
    \cline{1-8}
    \multirow{6}{*}{\tabincell{c}{email-euall\\(265,009, 364,481)}}& \multirow{2}{*}{2} & 12 & 412,779  & 1.61 & 1.34 &1.27 & \bf1.11  &  &  & 20 & 13,570,746  & 38.44 & 36.85 &38.68 &\bf28.89 \\
    \cline{10-16}
    & & 20 & 0  & 0.08 & \bf0.05 & \bf0.05 & \bf0.05 & &\multirow{3}{*}{3} & 12 & 3,306,582,222  & 10368.47 & 8910.28 &8894.69 &\bf5995.67 \\
    \cline{2-8}
     &\multirow{2}{*}{3} & 12 & 32,639,016  & 101.16 & 91.59 & 82.07 &\bf56.22 & & & 20 & 1,610,097,574  & 3950.16 & 5244.41 & 5325.80 &\bf3377.68 \\
    &  & 20 & 2637  & 0.34 & 0.30 & 0.21 &\bf0.19 &  &  & 30 & 4,626,307  & 32.36 &  76.44 & 49.43 & \bf30.10 \\
    \cline{2-8} \cline{10-16}
    &\multirow{2}{*}{4} & 12 & 1,940,182,978  & 7085.39 & 6041.73 & 5385.09 &\bf3535.63 & & 4 & 30 & 1,047,289,095  & 4167.85 & 10239.25 & 7556.32 &\bf4016.08 \\
    \cline{9-16} 
    & & 20 & 1,707,177  & 11.99 & 21.40 &13.30 & \bf7.70 & \multirow{7}{*}{\tabincell{c}{soc-pokec\\(1,632,803, 22,301,964)}} & \multirow{3}{*}{2} & 12 & 7,679,906  & 50.77 & 39.16 &36.78 & \bf35.53 \\
    \cline{1-8}
    \multirow{6}{*}{\tabincell{c}{com-dblp\\(317,080, 1,049,866)}} & \multirow{2}{*}{2} & 12 & 12,544  & 0.34 & \bf0.11 & \bf0.11 & \bf0.11 & & & 20 & 94,184  & \bf9.02 & 10.67 &10.24 &10.00 \\
     &  & 20 & 5049  & 0.27 & 0.06 &0.06 & \bf0.05 & & & 30 & 3  & 5.98 & 5.83 &5.72 & \bf5.46 \\
    \cline{2-8} \cline{10-16}
    &\multirow{2}{*}{3}  & 12 & 3,003,588  & 6.13 & 3.75 &3.68 & \bf3.51 & & \multirow{3}{*}{3} & 12 & 520,888,893  & 1719.85 & 1528.07 &1347.97 &\bf996.43 \\
    &  & 20 & 2,141,932  & 4.28 & 2.83 &2.77 & \bf2.57 & & & 20 & 5,911,456  & 30.52 & 39.38 &33.38 & \bf26.94 \\
    \cline{2-8}
    & \multirow{2}{*}{4} & 12 & 610,150,817  & 914.84 & 729.16 &720.13 & \bf666.98 & &  & 30 & 5  & 6.10 & 6.35 &6.47 & \bf5.92 \\
    \cline{10-16}
    &  &20  & 492,253,045  & 726.93 & 621.17 &612.59 & \bf546.30 &  & 4 & 20 & 318,035,938  & 1148.87 & 1722.87 &1292.65 & \bf780.34 \\
     \bottomrule[2pt]
    \end{tabular}
    }
    }
    \label{table::seq-comparison}
\end{table*}

\begin{figure*}[t]
    \centering
    \subfigure[wiki-vote $(k = 3)$]{\includegraphics[width=0.245\textwidth]{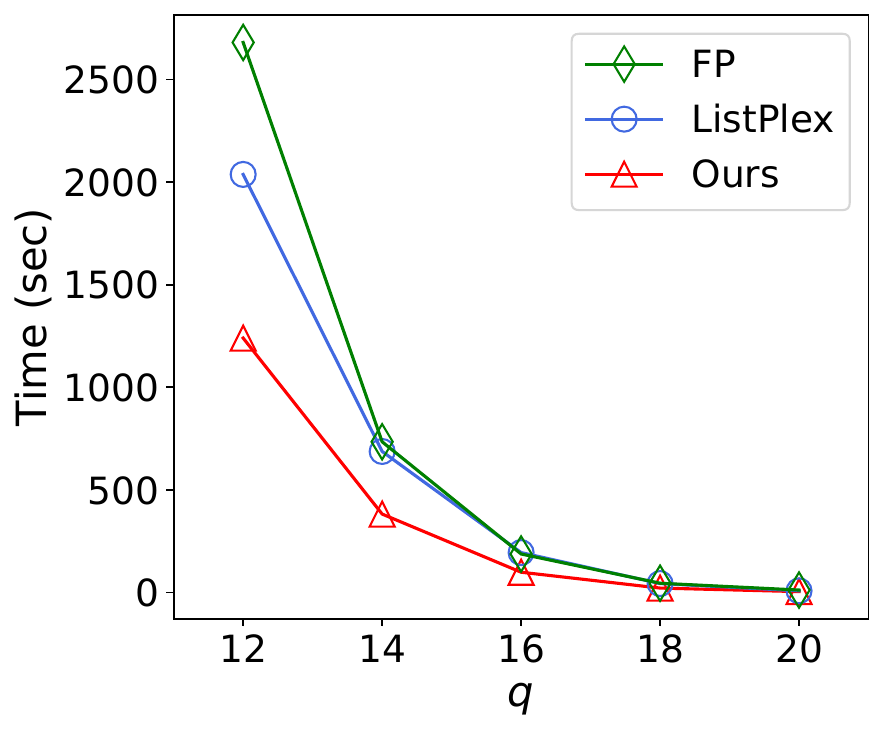}}
    \subfigure[wiki-vote $(k = 4)$]{\includegraphics[width=0.245\textwidth]{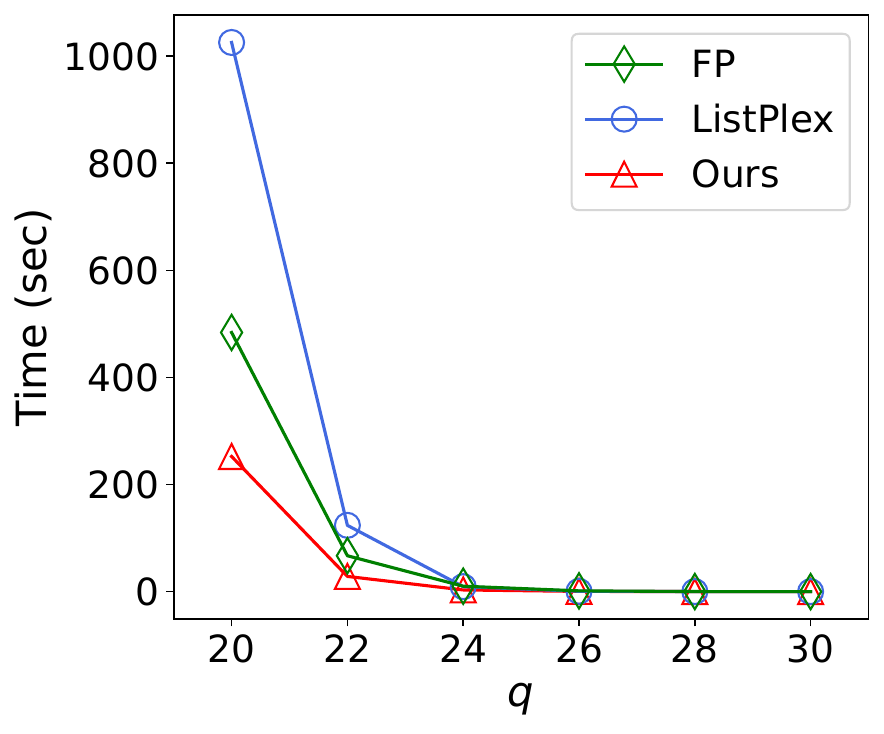}}
    \subfigure[soc-pokec $(k = 3)$]{\includegraphics[width=0.245\textwidth]{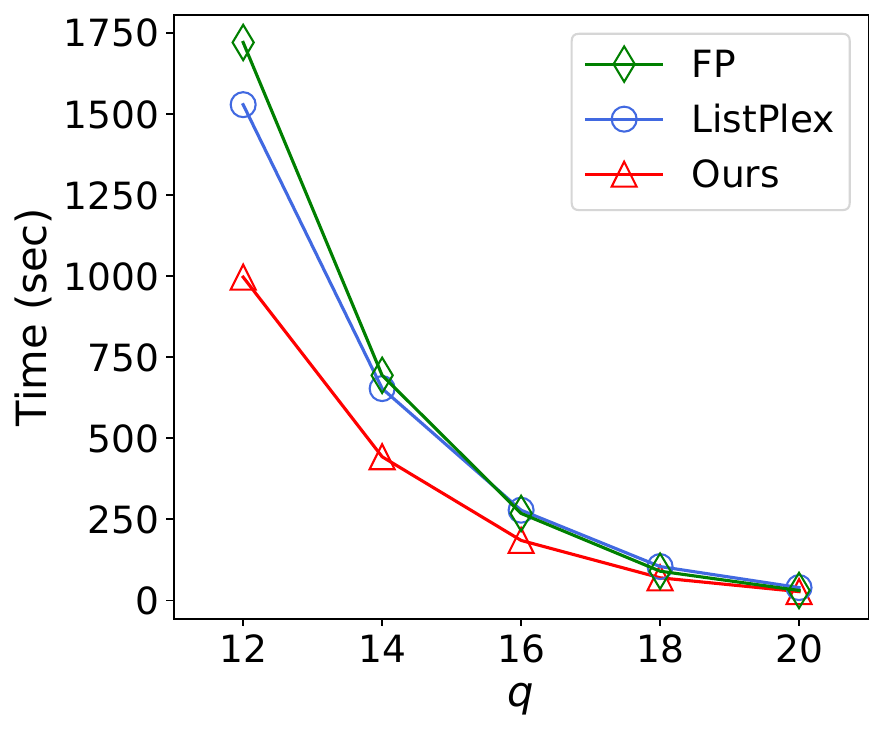}}
    \subfigure[soc-pokec $(k = 4)$]{\includegraphics[width=0.245\textwidth]{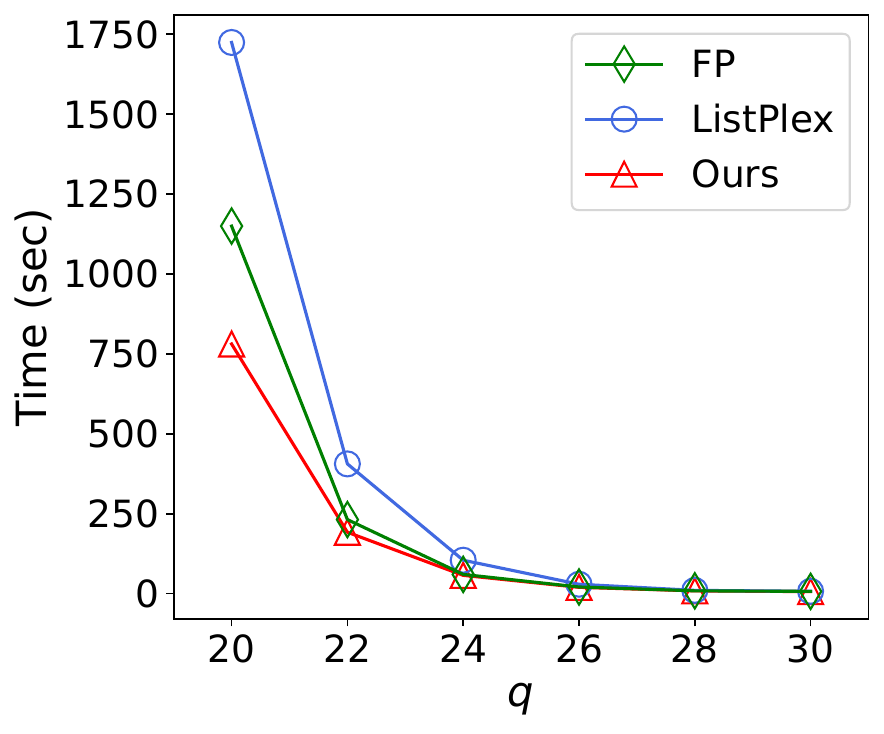}}
    \caption{Running Time (sec) of the Three Algorithms on Various Datasets and Parameters}
    \label{fig:different-q}
\end{figure*}

\vspace{1mm}
\noindent \textbf{Existing Methods for Comparison. } A few methods have been proposed for large maximal $k$-plex enumeration, including D2K~\cite{d2k}, CommuPlex~\cite{aaai2020maximal}, ListPlex~\cite{www22maximal}, and FP~\cite{cikm22maximal}. Among them, ListPlex~\footnote{\url{https://github.com/joey001/ListPlex}} and FP~\footnote{\url{https://github.com/qq-dai/kplexEnum}} outperform all the earlier works in terms of running time according to their experiments~\cite{cikm22maximal, www22maximal}, so they are chosen as our baselines for comparison. 
Note that ListPlex and FP are concurrent works, so there is no existing comparison between them. Thus, we choose these two state-of-the-art algorithms as baselines and compare our algorithm with them. 
Please refer to Section~\ref{sec:related} for a more detailed review of ListPlex and FP. 

In our description hereafter, efficiency means the running time unless otherwise specified. For all the tested algorithms, the running time includes the time for core decomposition, the time for subgraph construction, and the time for $k$-plex enumeration, but the graph loading time is excluded since it is a fixed constant.

Following the tradition of previous works, we set the $k$-plex size lower bound $q$ to be at least $(2k-1)$ which guarantees the connectivity of the $k$-plexes outputted. 

\vspace{1mm}
\noindent \textbf{Parameter Setting.} 
For experiments on sequential execution, we set parameters $k = 2, 3, 4$ and $q = 12, 20, 30$ following the parameter settings in~\cite{www22maximal, cikm22maximal}. Note that some parameter settings return no valid maximal $k$-plexes, while others lead to existing algorithms running for prohibitive amount of time, so we avoid reporting those settings. 
For as-skitter, using a small $q$ (e.g., 12, 20, 30) is too expensive so we use a larger value for $q$ following~\cite{www22maximal}. 

For experiments on parallel execution, using experiments that can finish quickly cannot justify the need for parallel execution. We, therefore, pick the settings of $q$ so that the job needs to run for some time to obtain quite some relatively large $k$-plexes. 

\vspace{1mm}
\noindent \textbf{Performance of Sequential Execution. } We first compare the sequential versions of our algorithm and the two baselines ListPlex and FP. 
Since sequential algorithms can be slow (esp.\ for the baselines), we use small and medium graphs. 

We have extensively tested our algorithm by comparing its outputs with those of ListPlex and FP on various datasets, and have verified that all three algorithms return the same result set for each dataset and parameters $(k, q)$.
Table~\ref{table::seq-comparison} shows the results where we can see that all three algorithms output the same number of $k$-plexes for each dataset and parameter $(k, q)$. In terms of time-efficiency, our algorithm outperforms ListPlex and FP except for rare cases that are easy (i.e., where all algorithms finish very quickly). 
For our two algorithm variants, Ours is consistently faster than Ours\_P, except for rare cases where both finish quickly in a similar amount of time. Moreover, our algorithm is consistently faster than ListPlex and FP. Specifically, Ours is up to $5 \times$ faster than ListPlex (e.g., soc-epinions, $k = 4$, $q = 30$), and up to $2 \times$ (e.g., email-euall, $k = 4$, $q = 12$) faster than FP, respectively. Also, there is no clear winner between ListPlex and FP: for example, ListPlex can be 3.56$\times$ slower than FP (e.g., as-skitter, $k = 3$, $q = 60$), but over 10$\times$ faster in other cases (e.g., as-skitter, $k = 3$, $q = 100$).

\renewcommand{\arraystretch}{1.2}
\begin{table*}[t]
    \centering
	\caption{Running Time (sec) of Listing Large Maximal $k$-Plexes on Large Graphs by Parallel Algorithms (16 Threads)}
    \label{tbl-community1}
    \resizebox{2.1\columnwidth}{!}{
    \begin{tabular}{c|c|c|c|c|c|c|c|c|c|c|c|c|c|c|c|c}
	\toprule[2pt]
    \multirow{2}{*}{\tabincell{c}{Network}} &\multirow{2}{*}{$n$} & \multirow{2}{*}{$m$}  &\multicolumn{7}{c|}{$k=2$} &\multicolumn{7}{c}{$k=3$}\\
    \cline{4-17}
    & & &$q$ &$\tau_{best}$ (ms) &\#$k$-plexes & FP & ListPlex & Ours & Ours ($\tau_{best}$) & $q$ &$\tau_{best}$(ms) &\#$k$-plexes & FP & ListPlex & Ours & Ours ($\tau_{best}$)\\
    \hline
    enwiki-2021 & 6,253,897 & 136,494,843 & 40 &0.01 & 1,443,280 &241.18 &291.22 &\bf154.99 & 151.01 & 50 &0.001 & 40,997 &19056.73 &3860.17 &\bf1008.26 & 1006.43\\
    \hline
    arabic-2005 & 2,2743,881 & 553,903,073 & 900 & 1 & 224,870,898 & 1873.71 & 417.91 & \bf388.42 & 385.52 & 1000 & 0.1 & 34,155,502 & 708.36 & 70.10 & \bf67.98 & 67.98\\
    \hline 
    uk-2005 & 39,454,463 & 783,027,125 & 250 & 0.01 & 159,199,947 & FAIL & 194.68 & \bf165.54 & 164.20 & 500 & 0.1 & 116,684,615 & 553.03 & 56.68 & \bf52.06 & 52.06 \\
    \hline
    it-2004 & 41,290,648 & 1,027,474,947 & 1000 &20 & 66,067,542 & 1958.70 & 2053.83 & \bf934.80 & 907.36 & 2000 &0.1 & 197,679,229 & 17785.82 & 458.83 & \bf401.13 & 401.13\\
    \hline
    webbase-2001 & 115,554,441 & 854,809,761 & 400 &0.1 & 59674227 & 222.81 & 67.45 & \bf60.93 &60.93 & 800 &0.1 & 1,785,341,050 & 15446.46 & 3312.95 & \bf3014.44 &3014.442\\ 
     \bottomrule[2pt]
    \end{tabular}
    }
\end{table*}

\begin{figure*}[ht]
    \centering
    \subfigure[$k=2$]
    {\includegraphics[width=0.48\textwidth]{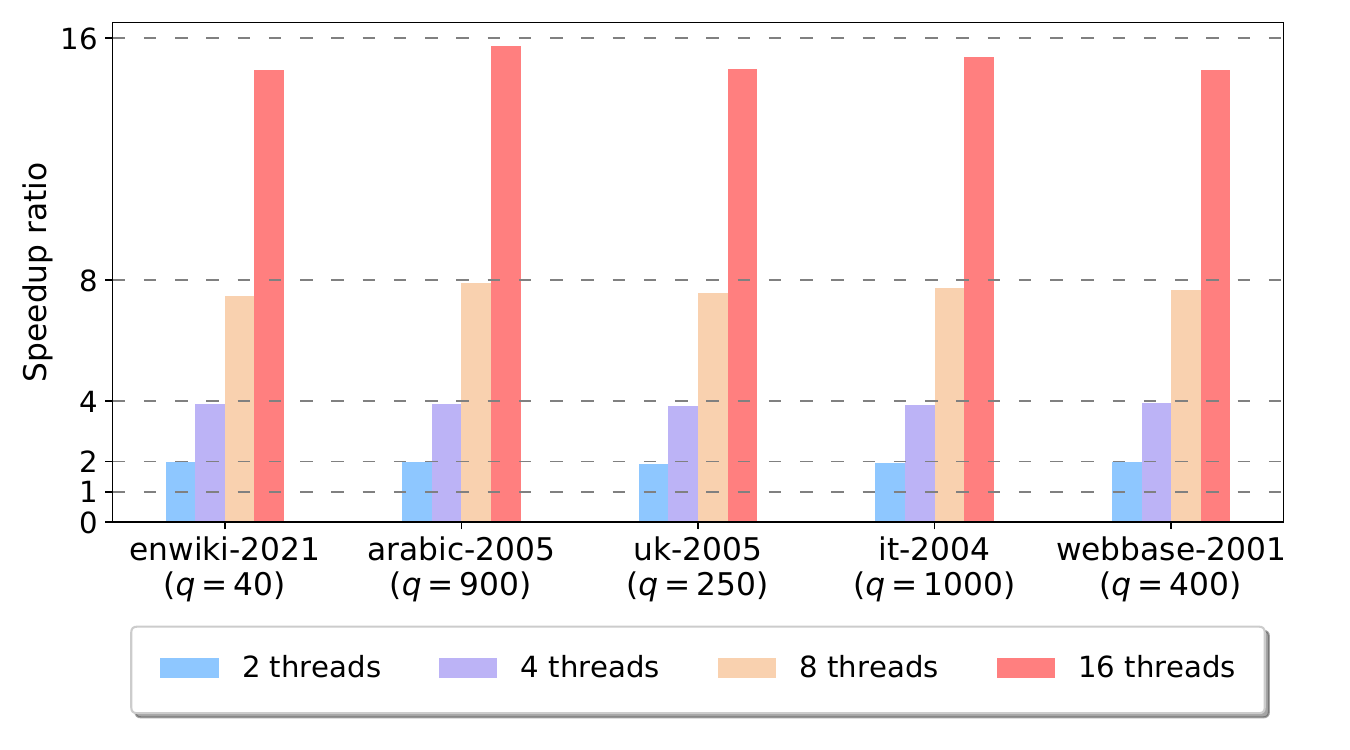}}
    \subfigure[$k=3$]
    {\includegraphics[width=0.48\textwidth]{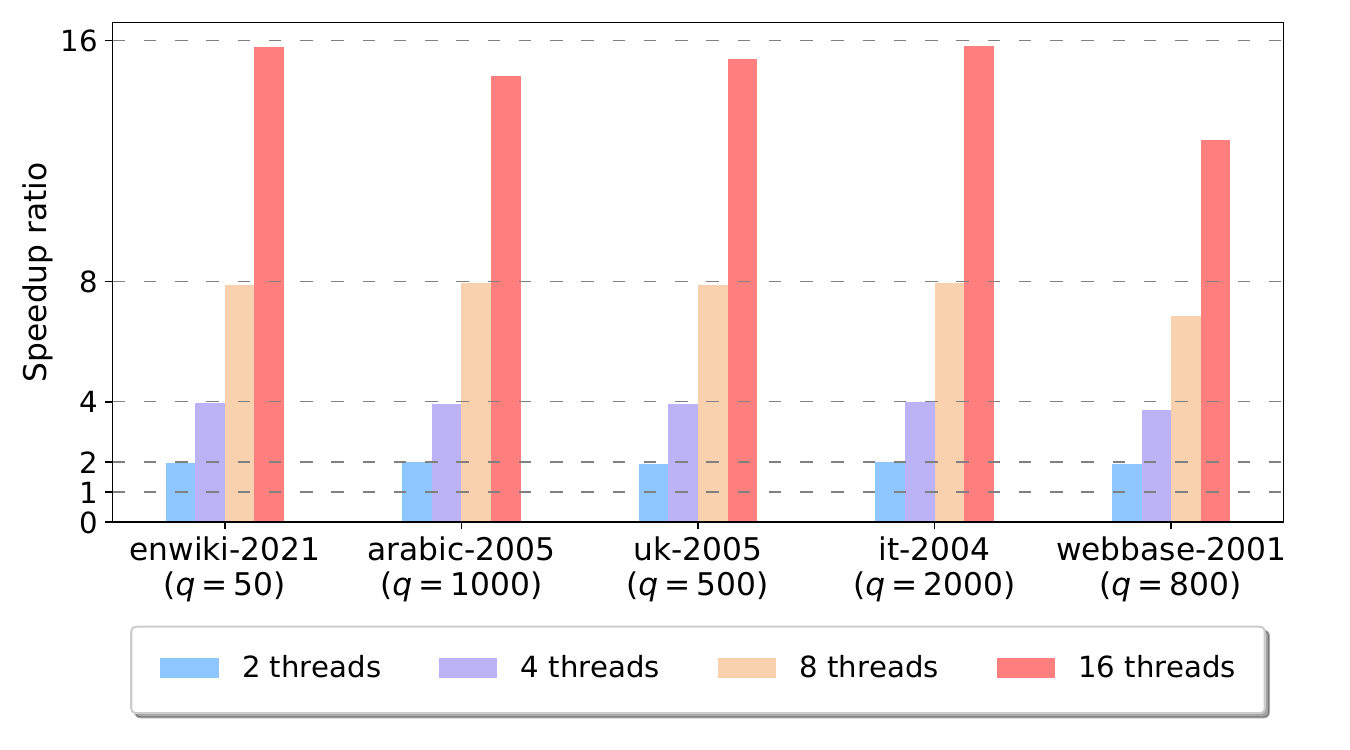}}
    \caption{Speedup Ratio of Our Parallel Algorithm in Five Large Graphs} 
    \label{fig-parallel}
\end{figure*}

We can see that our algorithm has more advantages in efficiency when the number of sub-tasks is larger. This is because our upper bounding and pruning techniques can effectively prune unfruitful sub-tasks. For instance, on dataset email-euall when $k = 4$ and $q = 12$, there are a lot of sub-tasks and the number of $k$-plexes outputted is thus huge; there, FP takes $7085.39$ seconds while Ours takes only $3535.63$ seconds. 

We also study how the performance of sequential algorithms changes when $q$ varies, and the results are shown in Figure~\ref{fig:different-q}. Due to space limitation, we only show results for four datasets, and more results can be found in Figure~\ref{fig:different-q-total} in Appendix~\ref{app::different-q-ablation-total}~\cite{fullpaper}. In each subfigure, the horizontal axis is $q$, and the vertical axis is the total running time. As Figure~\ref{fig:different-q} shows, Ours (red line) consistently uses less time than ListPlex and FP. For example, Ours is $4 \times$ faster than ListPlex on wiki-vote when $k = 4$, $q = 20$.

As for the performance between ListPlex and FP, we can see from Figure~\ref{fig:different-q} that when $k$ is small, ListPlex (blue line) is always faster than FP (green line) with different values of $q$. As $k$ becomes larger, FP can become faster than ListPlex. Note that the time complexity of ListPlex and FP are $O\left(n^{2 k}+n(D \Delta)^{k+1} \gamma_k^D\right)$ and $O\left(n^2 \gamma_k^n\right)$, respectively where $\gamma_k < 2$ is a constant. Therefore, when $k$ is small, the time complexity of ListPlex is smaller than FP; but as $k$ becomes large, the number of branches increases quickly and the upper-bounding technique in FP becomes effective. These results are also consistent with the statement in FP's paper~\cite{cikm22maximal}: {\em the speedup of FP increases dramatically with the increase of $k$}. As far as we know, this is the first time to compare the performance of ListPlex and FP, which are proposed in parallel very recently.

We notice that our reported running time of FP is slower than that reported in their paper~\cite{cikm22maximal} albeit the same settings, though our hardware is even more powerful. However, we double checked that we have faithfully run FP following their GitHub repo's instructions. 

We also compare the peak memory consumption of three algorithms. Please see Appendix~\ref{app::memory}~\cite{fullpaper} for the results. To summarize, ListPlex and Ours report similar memory usage, while FP needs more memory to store the intermediate results even just on medium datasets. For example, the memory usage of FP, ListPlex, and Ours on soc-pokec are 937.52~MB, 431.69~MB, and 431.26~MB.

\vspace{1mm}
\noindent \textbf{Performance of Parallel Execution.} We next compare the performance of the parallel versions of Ours, ListPlex, and FP, using the large graphs. Note that both ListPlex and FP provide their own parallel implementations, but they cannot eliminate straggler tasks like Ours, which adopts a timeout mechanism (c.f.\ Section~\ref{sec:parallel}). For Ours, we fix the timeout threshold $\tau_{time}=0.1$ ms by default to compare with parallel ListPlex and FP. We also include a variant ``Ours ($\tau_{best}$)'' which tunes $\tau_{time}$ to find its best value (i.e., $\tau_{best}$) that minimizes the running time for each individual dataset and each parameter pair $(k, q)$. 

Table~\ref{tbl-community1} shows the running time of parallel FP, ListPlex, Ours ($\tau_{time}=0.1$ ms) and Ours ($\tau_{time}=\tau_{best}$) running with 16 threads. Note that the tuned values of $\tau_{best}$ are also shown in Table~\ref{tbl-community1}, which vary in different test cases. 
Please refer to Appendix~\ref{app:tau}~\cite{fullpaper} for the experimental results on tuning $\tau_{time}$, where we can see that unfavorable values of $\tau_{time}$ (e.g., those too large for load balancing) may slow down the computation significantly. Overall, our default setting $\tau_{time}=0.1$ ms consistently performs very close to the best setting when $\tau_{time}=\tau_{best}$ in all test cases shown in Table~\ref{tbl-community1}, so it is a good default choice.

Compared with parallel ListPlex and FP, parallel Ours is significantly faster. For example, Ours is $18.9 \times$ and $3.8 \times$ faster than FP and ListPlex on dataset enwiki-2021 ($k=3$ and $q=50$), respectively. Note that FP is said to have {\em very high parallel performance}~\cite{cikm22maximal} and ListPlex also claims that it {\em can reach a nearly perfect speedup}~\cite{www22maximal}. Also note that FP fails on uk-2005 when $k=2$ and $q=250$, likely due to a bug in the code implementation of parallel FP. In fact, FP can be a few times slower than ListPlex, since its parallel implementation does not parallelize the subgraph construction step: all subgraphs are constructed in serial at the beginning which can become the major performance bottleneck.

We also evaluate the scale-up performance of our parallel algorithm. Figure~\ref{fig-parallel} shows the speedup results, where we can see that Ours scales nearly ideally with the number of threads on all the five large datasets for all the tested parameters $k$ and $q$ used in Table~\ref{tbl-community1}. For example, on dataset it-2004 ($k = 3$ and $q = 2000$), it achieves $7.93 \times$ and $15.82 \times$ speedup with $8$ and $16$ threads, respectively.

\vspace{1mm}
\noindent \textbf{Ablation Study.} We now conduct ablation study to verify the effectiveness of our upper-bound-based pruning technique as specified in Lines~\ref{alg2:line17}-\ref{alg2:line18} of Algorithm~\ref{alg::2}, where the upper bound is computed with Eq~(\ref{eq:our_ub}). While ListPlex does not apply any upper-bound-based pruning, FP uses one that requires a time-consuming sorting procedure in the computation of upper bound (c.f., Lemma~5 of~\cite{cikm22maximal}). 

\renewcommand{\arraystretch}{1.0}
\begin{table}[t]
    \centering
	\caption{Effect of Different Upper Bounding Techniques}
	\label{tbl-community4}
    \resizebox{0.85\columnwidth}{!}{
    \begin{tabular}{c|c|c|c|c|c}
	\toprule[1pt]
	\multirow{2}{*}{\tabincell{c}{Network}} & \multirow{2}{*}{$k$}    & \multirow{2}{*}{$q$} &  \multicolumn{3}{c}{ Running time (sec)}
	\\
	\cline{4-6}
	& &  & Ours\textbackslash ub & Ours\textbackslash ub+fp & Ours \\
     \hline \multirow{4}{*}{\tabincell{c}{wiki-vote}} &\multirow{2}{*}{3} &12 &1393.50 &1319.05 &\bf1239.83 \\
    & & 20 &5.20 &4.72 &\bf4.15 \\
    \cline{2-6}
    &\multirow{2}{*}{4} &20 &530.48 &280.75 &\bf252.40 \\
    & & 30 &0.14 &0.13 &\bf0.06 \\
    \hline \multirow{4}{*}{\tabincell{c}{soc-epinions}} &\multirow{2}{*}{2} & 12 &138.82 &142.06 &\bf 130.14 \\
    & & 20 &14.92 &15.48 &\bf 14.01 \\
    \cline{2-6}
    &\multirow{2}{*}{3} & 20 &1699.49 &1687.29 &\bf 1540.87 \\
    & & 30 &2.87 &2.44 &\bf 2.11 \\
    \hline \multirow{4}{*}{\tabincell{c}{email-euall}} &\multirow{2}{*}{3} &12 &62.85 &63.83 &\bf 56.22 \\
    & & 20 &0.29 &0.28 &\bf 0.19 \\
    \cline{2-6}
    &\multirow{2}{*}{4} & 12 &4367.88 &3961.40 &\bf3535.63 \\
    & & 20 &13.01 &9.31 &\bf7.70 \\
    \hline \multirow{4}{*}{\tabincell{c}{soc-pokec}} &\multirow{2}{*}{3} &12 &1039.61 &1022.14 &\bf996.43 \\
    & & 20 &27.21 &29.19 &\bf26.94 \\
    \cline{2-6}
    &\multirow{2}{*}{4} &20 &988.90 &877.95 &\bf780.34 \\
    & & 30 &6.91 &6.76 &\bf6.37 \\
    
    \bottomrule[1pt]
    \end{tabular}
    }
\end{table}

\begin{figure*}[htbp]
    \centering
    \subfigure[soc-epinions $(k = 3)$]{\includegraphics[width=0.245\textwidth]{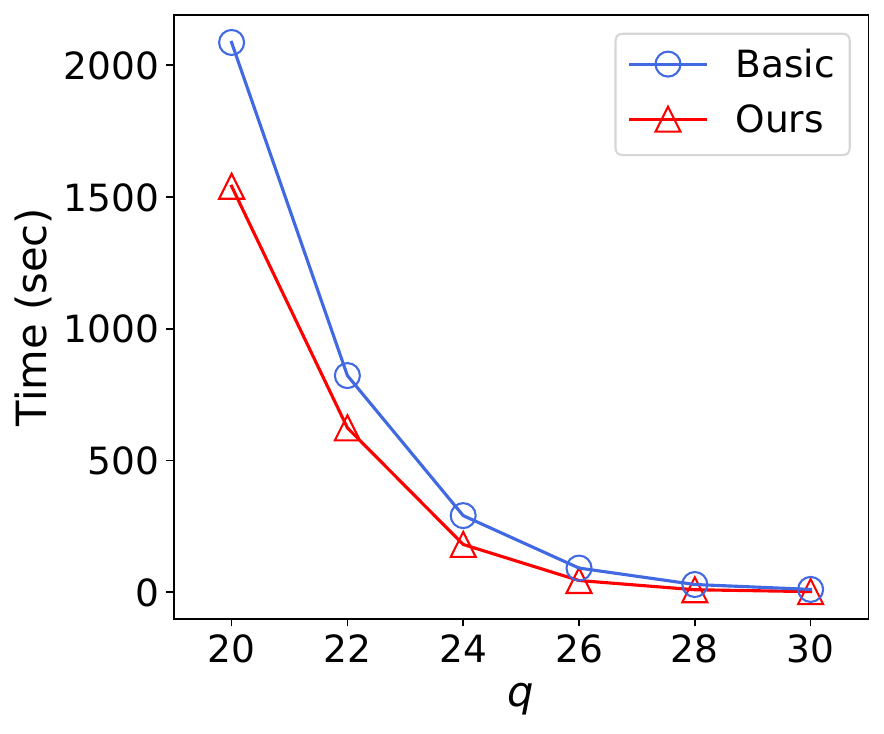}}
    \subfigure[email-euall $(k = 4)$]{\includegraphics[width=0.245\textwidth]{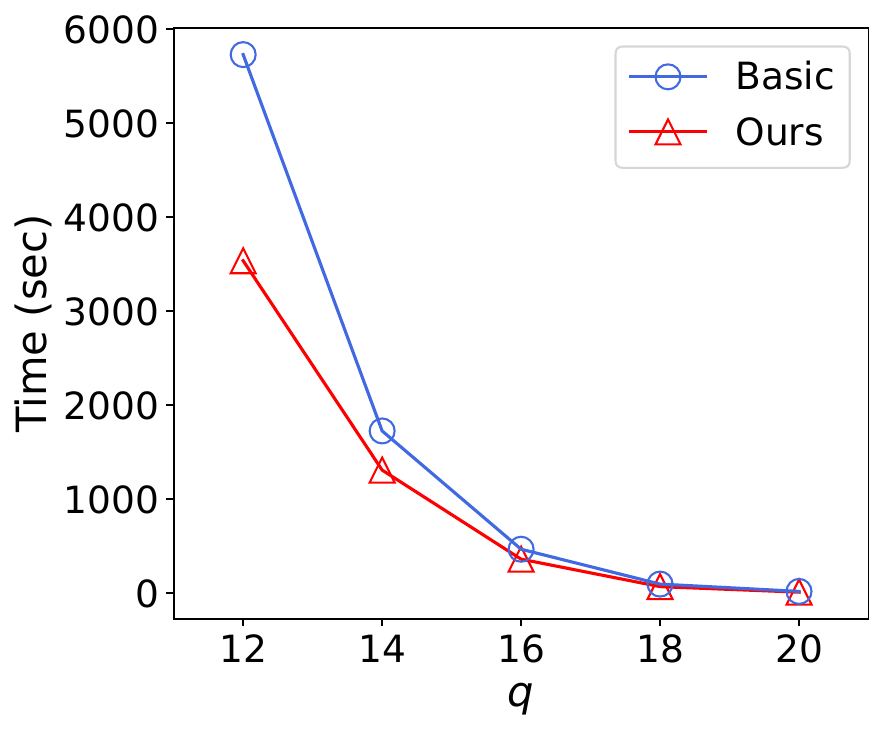}}
    \subfigure[soc-pokec $(k = 3)$]{\includegraphics[width=0.245\textwidth]{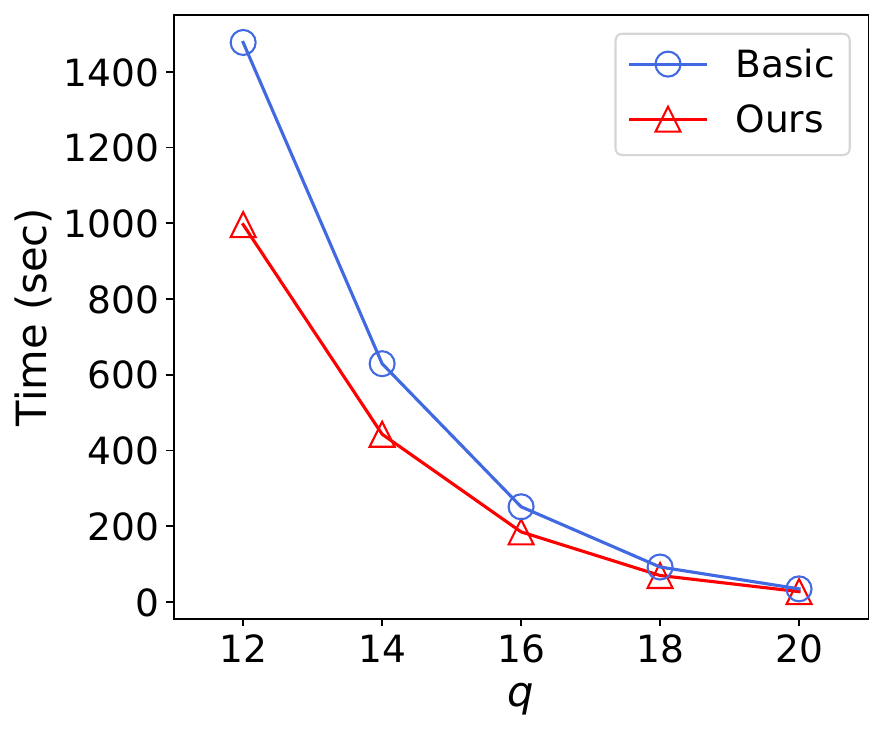}}
    \subfigure[wiki-vote $(k = 4)$]{\includegraphics[width=0.245\textwidth]{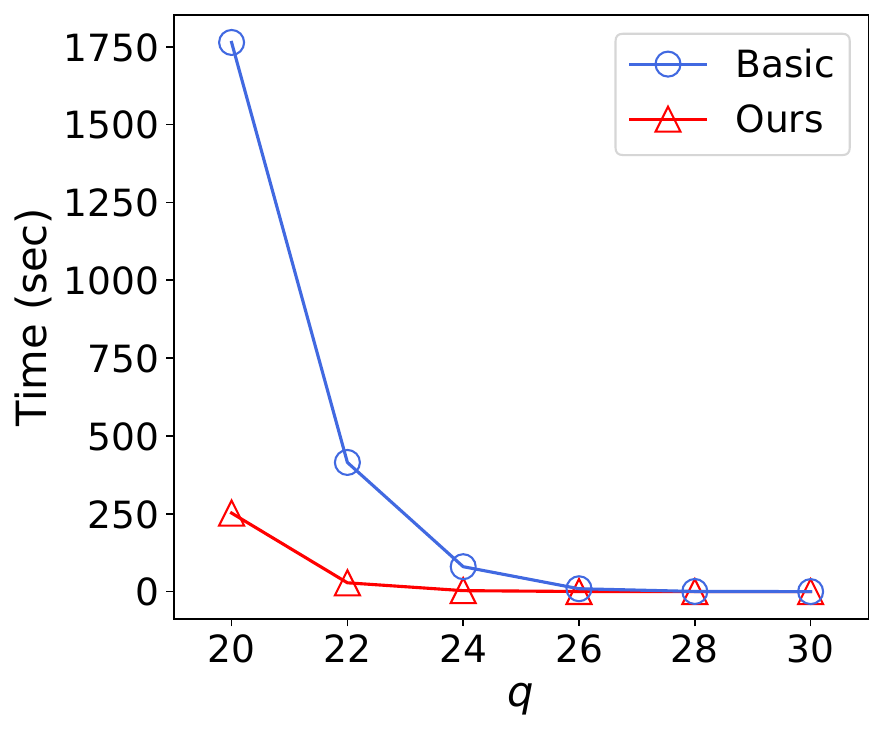}}
    \caption{Running time (sec) of Our Basic and Optimized Algorithms on Various Datasets and Parameters}
    \label{fig:different-q-ablation}
\end{figure*}

The ablation study results are shown in Table~\ref{tbl-community4}, where we use ``Ours\textbackslash ub'' to denote our algorithm variant without using upper-bound-based pruning, and use ``Ours\textbackslash ub+fp'' to denote our algorithm that directly uses the upper bounding technique of FP~\cite{cikm22maximal} instead. In Table~\ref{tbl-community4}, we show the results on four representative datasets with different $k$ and $q$ (the results on other datasets are similar and omitted due to space limit). We can see that Ours outperforms ``Ours\textbackslash ub'' and ``Ours\textbackslash ub+fp'' in all the cases. This shows that while using upper-bound-based pruning improves performance in this framework, FP's upper bounding technique is not as effective as Ours, due to the need of costly sorting when computing the upper bound in each recursion. In fact, ``Ours\textbackslash ub+fp'' can be even slower than ``Ours\textbackslash ub'' (c.f., soc-epinions with $k=2$ and $q=12$) since the expensive sorting procedure in the computation of upper bound backfires, while the upper-bound-based pruning does not reduce the branches much. Another observation is that our upper-bounding technique is more effective when $k$ and the number of sub-tasks become larger (e.g., when $q$ is smaller). For example, the running time of Ours and ``Ours\textbackslash ub'' is $252.40$ seconds and $530.48$ seconds, respectively, on dataset wiki-vote with $k = 4$ and $q = 20$.

\begin{table}[t]
    \centering
	\caption{Effect of Pruning Rules}
	\label{tbl-community5}
    \resizebox{\columnwidth}{!}{
    \begin{tabular}{c|c|c|c|c|c|c}
	\toprule[1pt]
	\multirow{2}{*}{\tabincell{c}{Network}} & \multirow{2}{*}{$k$}    & \multirow{2}{*}{$q$} &  \multicolumn{4}{c}{ Running time (sec)}
	\\
	\cline{4-7}
	& &  & Basic  & Basic+R1 & Basic+R2 & Ours  \\
     \hline \multirow{4}{*}{\tabincell{c}{wiki-vote}} &\multirow{2}{*}{3} &12 &2124.54 &1726.29 &1269.16 &\bf1239.83 \\
    & & 20 &30.93 &16.09 &4.54 &\bf4.15 \\
    \cline{2-7}
    &\multirow{2}{*}{4} &20 &1763.94 &791.73 &272.78 &\bf252.40 \\
    & & 30 &0.14 &0.15 &0.15 &\bf0.06 \\
    \hline \multirow{4}{*}{\tabincell{c}{soc-epinions}} &\multirow{2}{*}{2} & 12 &148.09 &145.56 &135.93 &\bf130.14 \\
    & & 20 &16.42 &16.31 &14.86 &\bf14.01 \\
    \cline{2-7}
    &\multirow{2}{*}{3} & 20 &2086.60 &1796.44 &1582.69 &\bf1540.87 \\
    & & 30 &10.48 &6.43 &2.30 &\bf2.11 \\
    \hline \multirow{4}{*}{\tabincell{c}{email-euall}} &\multirow{2}{*}{3} &12 &83.05 &73.29 &60.51 &\bf56.22 \\
    & & 20 &0.53 &0.44 &0.28 &\bf0.19 \\
    \cline{2-7}
    &\multirow{2}{*}{4} & 12 &5729.98 &4997.78 &3708.91 &\textbf{3535.63} \\
    & & 20 &15.39 &11.58 &8.04 &\bf7.70 \\
    \hline \multirow{4}{*}{\tabincell{c}{soc-pokec}} &\multirow{2}{*}{3} &12 &1477.74 &1272.26 &1009.03 &\bf996.43 \\
    & & 20 &34.25 &30.09 &27.07 &\bf26.94 \\
    \cline{2-7}
    &\multirow{2}{*}{4} &20 &1003.68 &886.50 &791.87 &\bf780.34 \\
    & & 30 &6.88 &6.57 &6.63 &\bf6.37 \\
    
    \bottomrule[1pt]
    \end{tabular}
    }
\end{table}

We next conduct ablation study to verify the effectiveness of our pruning rules, including (R1)~Theorem~\ref{reduction4} for pruning initial sub-tasks right before Line~\ref{algo1:line10} of Algorithm~\ref{alg::1}, and (R2)~Theorems~\ref{reduction1}, ~\ref{reduction3} and~\ref{reduction2} for second-order-based pruning to shrink the candidate and exclusive sets during recursion.

The ablation study results are shown in Table~\ref{tbl-community5}, where our algorithm variant without R1 and R2 is denoted by ``Basic''. In Table~\ref{tbl-community5}, we can see that both R1 and R2 bring performance improvements on the four tested graphs. The pruning rules are the most effective on dataset wiki-vote with $k = 4$ and $q = 20$, where Ours achieves $7 \times$ speedup compared with Basic.

Figure~\ref{fig:different-q-ablation} further compares the running time between Basic and Ours as $k$ and $q$ vary (where more values of $q$ are tested). Due to space limitation, we only show results for four datasets, and more results can be found in Figure~\ref{fig:different-q-ablation-total} in Appendix~\ref{app::basic}~\cite{fullpaper}. We can see that Ours is consistently faster than the basic version with different $k$ and $q$. This demonstrates the effectiveness of our pruning rules.

\section{Conclusion}\label{sec:conclude}
In this paper, we proposed an efficient branch-and-bound algorithm to enumerate all maximal $k$-plexes with at least $q$ vertices. Our algorithm adopts an effective search space partitioning approach that provides a good time complexity, a new pivot vertex selection method that reduces candidate size, an effective upper-bounding technique to prune useless branches, and three novel pruning techniques by vertex pairs. Our parallel algorithm version uses a timeout mechanism to eliminate straggler tasks. Extensive experiments show that our algorithms compare favorably with the state-of-the-art algorithms, and the performance of our parallel algorithm version scales nearly ideally with the number of threads.

\begin{acks}
This work was supported by DOE ECRP Award 0000274975, NSF OIA-2229394, NSF OAC-2414474, and NSF OAC-2414185. 
Additionally, Zhongyi Huang was partially supported by the NSFC Project No. 12025104. 
\end{acks}

\bibliographystyle{ACM-Reference-Format}
\bibliography{sample-base}

%
% ...

\appendix

\clearpage
\begin{appendix}
\section{Proofs of Theorems and Lemmas}
\label{proofs}

\subsection{Proof of Theorem~\ref{lemma::2nd_order}}\label{app:th1}
\begin{proof} This can be seen from Figure~\ref{2nd_order}. Let us first define $\overline{N^*_P}(v) = \overline{N_P}(v) - \{v\}$, so $|\overline{N^*_P}(v)| \leq k-1$. 
In Case~(i) where $(u, v)\not\in E$, any vertex $w\in P$ can only fall in the following 3 scenarios: (1)~$w\in\overline{N^*_P}(u)$, (2)~$w\in\overline{N^*_P}(v)$, and (3)~$w\in N_P(u)\cap N_P(v)$. Note that $w$ may be in both (1) and (2). So we have:
\begin{eqnarray}
|P| & = & |N_P(u)\cap N_P(v)| + |\overline{N^*_P}(u)\cup\overline{N^*_P}(v)|\nonumber\\
& \leq & |N_P(u)\cap N_P(v)| + |\overline{N^*_P}(u)| + |\overline{N^*_P}(v)|\nonumber\\
& \leq & |N_P(u)\cap N_P(v)| + 2(k-1)\nonumber\\
& = & |N_P(u)\cap N_P(v)| + 2k-2,\nonumber
\end{eqnarray}
so $|N_P(u)\cap N_P(v)|\ge |P|-2k+2\ge q-2k+2$.

In Case~(ii) where $(u, v)\in E$, any vertex $w\in P$ can only be in one of the following 4 scenarios: (1)~$w=u$, (2)~$w=v$, (3)~$w\in N_P(u)\cap N_P(v)$, and (4)~$w\in\overline{N_P}(u)\cup\overline{N_P}(v)$, so:
\begin{eqnarray}
|P| & = & 2 + |N_P(u)\cap N_P(v)| + |\overline{N^*_P}(u)\cup\overline{N^*_P}(v)|\nonumber\\
& \leq & 2 + |N_P(u)\cap N_P(v)| + |\overline{N^*_P}(u)| + |\overline{N^*_P}(v)|\nonumber\\
& \leq & 2 + |N_P(u)\cap N_P(v)| + 2(k-1)\nonumber\\
& = & |N_P(u)\cap N_P(v)| + 2k,\nonumber
\end{eqnarray}
so $|N_P(u)\cap N_P(v)|\ge |P|-2k\ge q-2k$. \end{proof}

\begin{figure}[ht]
\centering
\includegraphics[width=0.9\columnwidth]{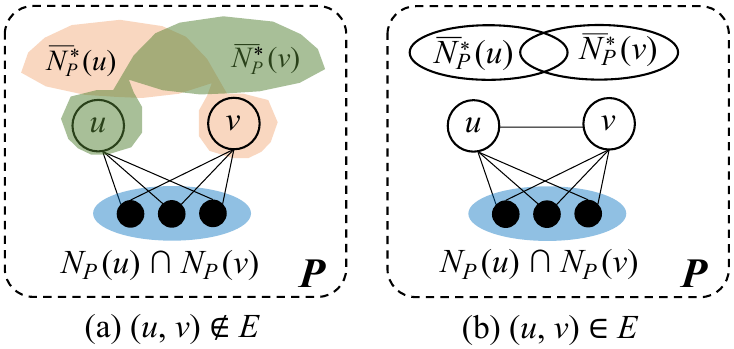}
\caption{Second-Order Pruning}\label{2nd_order}
\end{figure}
\setlength{\textfloatsep}{5pt}



\subsection{Proof of Theorem~\ref{lemma::bound4}}\label{app:th3}

\begin{proof}
We prove this by contradiction. Assume that there exists $K'=P_{m}\cap N_C(v_p)$ with $|K'|>|K|$. Also, let us denote by $\psi=\{w_1, w_2, \ldots, w_\ell\}$ the vertex ordering of $w\in N_C(v_p)$ in Line~\ref{alg3:line5} to create $K$, as illustrated in Figure~\ref{supp}.

Specifically, Figure~\ref{supp} top illustrates the execution flow of Lines~\ref{alg3:line4}--\ref{alg3:line8} in Algorithm~\ref{alg::2}, where $w_1$, $w_3$ and $w_4$ select their non-neighbor $u_1$ as $u_m$ in Line~\ref{alg3:line5}, $w_2$ and $w_5$ select $u_2$ as $u_m$, and $w_6$ selects $u_3$ as $u_m$. We define $\{w_1, w_3, w_4\}$ as $u_1$-group, $\{w_2, w_5\}$ as $u_2$-group, and $\{w_6\}$ as $u_3$-group.

\begin{figure}[h]
\centering
\includegraphics[width=0.86\columnwidth]{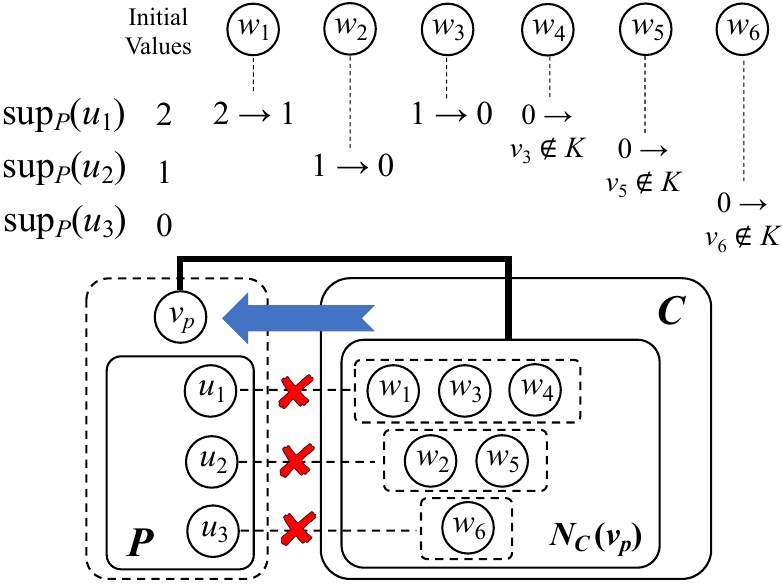}
\caption{Illustration of the Proof of Theorem~\ref{lemma::bound4}}
\label{supp}
\end{figure}
\setlength{\textfloatsep}{5pt}

Let us consider the update of $\text{sup}_P(u_1)$, whose initial value computed by Line~\ref{alg3:line2} is assumed to be $2$. When processing $w_1$, Line~\ref{alg3:line7} decrements $\text{sup}_P(u_1)$ as 1. Then, $w_3$ decrements it as 0. When processing $w_4$, since Line~\ref{alg3:line4} already finds that $\text{sup}_P(u_1)=0$, $w_4$ is excluded from $K$ (i.e., Line~\ref{alg3:line8} does not add it to $ub$). In a similar spirit, $w_5\notin K$ since $\text{sup}_P(u_2)=0$, and $w_6\notin K$ since $\text{sup}_P(u_3)=0$.

Note that since $w_1$, $w_3$ and $w_4$ cannot co-exist in a $k$-plex containing $P\cup \{v_p\}$, they cannot all belong to $K'$. In other words, if $w_4\notin K$ belongs to $K'$, then at least one of $w_1$ and $w_3$ is not in $K'$. In a similar spirit, if $w_5\notin K$ belongs to $K'$, then $w_2\not\in K'$. As for those $u_i$ whose initial value of $\text{sup}_P(u_i)$ is 0, we can show that any element in $u_i$-group can belong to neither $K$ nor $K'$. See $u_3$-group $=\{w_6\}$ in Figure~\ref{supp} for example. This is because if $w$ is added to $P$, then $\overline{d_P}(u_i)>k$ so $P$ cannot be a $k$-plex.

In general, in each $u_i$-group where $\text{sup}_P(u_i)\neq 0$, if a vertex $w\in K'-K$ exists (e.g., $w_4$ in Figure~\ref{supp}), then there must exist a different vertex $w'\in K-K'$ in the $u_i$-group (e.g., $w_1$ or $w_3$). This implies that $|K-K'|\geq |K'-K|$.

Therefore, we have
\begin{eqnarray*}
|K| & = & |K-K'| + |K\cap K'|\\
& \ge & |K'-K| + |K\cap K'|\ =\ |K'|,
\end{eqnarray*}
which contradicts with our assumption $|K'|>|K|$.
\end{proof}
\subsection{Proof of Theorem~\ref{reduction4}}\label{app:th4}
\begin{proof}
%
Let $P_{m}\supseteq P_S=S\cup\{v_i\}$ be a maximum $k$-plex containing $P_S$, then following the proof of Theorem~\ref{lemma::bound4}, we divide $P_m$ into three disjoint sets: (1)~$P_S$, (2)~$P_{m}\cap\overline{N_{C_S}}(v_i)$, (3)~$P_{m}\cap N_{C_S}(v_i)$. Note that $P_{m}\cap\overline{N_{C_S}}(v_i)=\emptyset$ since $\overline{N_{C_S}}(v_i)=\emptyset$ (as $C_S=N_{G_i}(v_i)$). Also, Theorem~\ref{lemma::bound4} has proved that $|P_{m}\cap N_{C_S}(v_i)|\leq |K|$. Therefore, we have $|P_m|\leq |P_S| + 0 + |K|=|P_S| + |K|$.
\end{proof}
\subsection{Proof of Lemma~\ref{th:complexity_alg3}}\label{app:th5}

\begin{proof}
To implement Algorithm~\ref{alg::3}, given a seed graph $G_i=(V_i, E_i)$, we maintain $d_P(v)$ for every $v\in V_i$. These degrees $d_P(.)$ are incrementally updated; for example, when a vertex $v_p$ is moved into $P$, we will increment $d_P(v)$ for every $v\in N_{G_i}(v_p)$. As a result, we can compute $\overline{d_P}(v)=|P|-d_P(v)$ and $\text{sup}_P(v)=k-\overline{d_P}(u)$ in $O(1)$ time. 

Moreover, we materialize $\text{sup}_P(u)$ for each vertex $u\in P$ in Line~\ref{alg3:line2}, so that in Line~\ref{alg3:line5} we can access them directly to compute $u_m$, and Line~\ref{alg3:line7} can be updated in $O(1)$ time.

Since $P$ is a $k$-plex of $G_i$, $|P|$ is bounded by $O(D+k)$ by Theorem~\ref{lemma::bound2}. Thus, Line~\ref{alg3:line2} takes $O(D+k)$ time. 

Also, $|N_{C}(v_p)|$ is bounded by $O(D)$ since $N_{C}(v_p) \subseteq C_S = N_{G_i}(v_i)$, so the for-loop in Line~\ref{alg3:line4} is executed for $O(D)$ iterations. In each iteration, Line~\ref{alg3:line5} takes $O(k)$ time since $|\overline{N_P}(w)|\leq k$, so the entire for-loop in Line~\ref{alg3:line4}--\ref{alg3:line8} takes $O(kD)$.

Putting them together, the time complexity of Algorithm~\ref{alg::3} is $O(D+k) + O(kD) = O(k+(k+1)D)\approx O(D)$ as $k$ is usually very small constant.
\end{proof}

\subsection{Proof of Lemma~\ref{th:subtasks}}\label{app:th6}

\begin{proof}
Since $\eta=\{v_1,\dots, v_{n}\}$ is the degeneracy ordering of $V$ and $V_i\subseteq\{v_i,v_{i+1},\dots,v_{n}\}$, we have $d_{G_i}(v_i)=|N_{G_i}(v_i)|\leq D$.

To show that $|N_{G_i}^2(v_i)|=O\left(\frac{D\Delta}{q-2k+2}\right)$, consider $E^*=\{(v, u)\,|\,v\in N_{G_i}(v_i)\ \wedge\ u\in N^2_{G_i}(v_i)\}$, which are those edges between $N_{G_i}(v_i)$ and $N^2_{G_i}(v_i)$ in Figure~\ref{N2}. Since $|N_{G_i}(v_i)|\leq D$, and each $v\in N_{G_i}(v_i)$ has at most $\Delta$ neighbors in $N^2_{G_i}(v_i)$, we have $|E^*|\leq D\Delta$. Also, let us denote by $E'$ all those edges $(v, u)\in E^*$ that are valid (i.e., $v$ and $u$ can appear in a $k$-plex $P$ in $G_i$ with $|P|\geq q$), then $|E'|\leq|E^*|\leq D\Delta$.

Recall from Corollary~\ref{corollary:2nd_order} that if $u\in N^2_{G_i}(v_i)$ belongs to a valid $k$-plex in $G_i$, then $|N_{G_i}(u)\cap N_{G_i}(v_i)|\geq q-2k+2$. This means that each valid $u\in N^2_{G_i}(v_i)$ share with $v_i$ at least $(q-2k+2)$ common neighbors that are in $N_{G_i}(v_i)$ (c.f., Figure~\ref{N2}), or equivalently, $u$ is adjacent to (or, uses) at least $(q-2k+2)$ edges $(v, u)$ in $E'$. Therefore, the number of valid $u\in N^2_{G_i}(v_i)$ is bounded by $\frac{|E'|}{q-2k+2}\leq\frac{D\Delta}{q-2k+2}$.

It may occur that $\frac{D\Delta}{q-2k+2}>n$, in which case we use $|N_{G_i}^2(v_i)|=O(n)$ instead. Combining the above two cases, we have $|N_{G_i}^2(v_i)|=O(r_1)$ where $r_1=\min\left\{\frac{D\Delta}{q-2k+2},n\right\}$.

Finally, the number of subsets $S\subseteq N_{G_i}^2(v_i)$ ($|S|\leq k-1$) (c.f., Line~\ref{algo1:line7} of Algorithm~\ref{alg::1}) is bounded by
$C_{r_1}^0+C_{r_1}^1 + \cdots + C_{r_1}^{k-1}\approx O\left(r_1^k\right)$,
since $k$ is a small constant.
\end{proof}

\subsection{Proof of Theorem~\ref{th:complexity}}\label{app:th7}

\begin{proof}
We just showed that the recursive body of Branch(.) takes time $O(|P|(|C|+|X|))$. Let us first bound $|X|$. Recall Algorithm~\ref{alg::1}, where by Line~\ref{algo1:line9}, vertices of $X\subseteq X_S$ are from either $V'_i$ or $(N^2_{G_i}(v_i)-S)$. Moreover, by Line~\ref{algo1:line5}, vertices of $V'_i$ are from either $N_G(v_i)$ or $N^2_G(v_i)$. Since $(N^2_{G_i}(v_i)-S)\subseteq N^2_G(v_i)$, vertices of $X$ are from either $N_G(v_i)$ or $ N^2_G(v_i)$. Let us denote $X_1=X\cap N_G(v_i)$ and $X_2=X\cap N^2_G(v_i)$.

We first bound $X_2$.  Consider $E^*=\{(u, v)\,|\,u\in N_{G}(v_i)\ \wedge\ v\in X_2\}$. Since $|N_{G}(v_i)|\leq \Delta$, and each $v\in N_{G}(v_i)$ has at most $\Delta$ neighbors in $N^2_{G}(v_i)$, we have $|E^*|\leq\Delta^2$. Recall from Theorem~\ref{lemma::2nd_order} that if $v\in X_2$ belongs to $k$-plex $P$ with $|P|\geq q$, then $|N_{G}(v)\cap N_{G}(v_i)|\geq q-2k+2$. This means that each $v\in X_2$ share with $v_i$ at least $(q-2k+2)$ common neighbors that are in $N_{G}(v_i)$, or equivalently, $v$ is adjacent to (or, uses) at least $(q-2k+2)$ edges $(u, v)$ in $E^*$. Therefore, the number of $v\in X_2$ is bounded by $\frac{|E^*|}{q-2k+2}\leq\frac{\Delta^2}{q-2k+2}$.

As for $X_1\subseteq  N_G(v_i)$, we have $|X_1|\leq \Delta$. 
In general, we do not set $q$ to be too large in reality, or there would be no results, so $(q-2k+2)$ is often much smaller than $\Delta$. 
Therefore, $|X|=|X_1|+|X_2|=O(\frac{\Delta^2}{q-2k+2} + \Delta)\approx O(\frac{\Delta^2}{q-2k+2})$.

Since $|P|$ is bounded by $O(D+k)\approx O(D)$ by Theorem~\ref{lemma::bound2}, and $C\subseteq C_S$ so $|C|\leq D$, the recursive body of Branch(.) takes $O(|P|(|C|+|X|))\approx O\left(D(D+\frac{\Delta^2}{q-2k+2})\right)\approx O\left(\frac{D\Delta^2}{q-2k+2}\right)$ time. This is because $\frac{\Delta^2}{q-2k+2} > \Delta \geq D$.

It may occur that $\frac{\Delta^2}{q-2k+2} > n$, in which case we use $|X|=O(n)$ instead, so the recursive body of Branch(.) takes $O(|P|(|C|+|X|))\approx O\left(D(D+n)\right)=O(nD)$ time.

Combining the above two cases, the recursive body of Branch(.) takes $O(r_2)$ time where $r_2 = \min\left\{\frac{D\Delta^2}{q-2k+2},nD\right\}$.

By Lemma~\ref{th:times}, Branch$(G_i,k,q,P_S,C_S,X_S)$ in Line~\ref{algo1:line10} of Algorithm~\ref{alg::1} recursively calls the body of Algorithm~\ref{alg::3} for $O(\gamma_k^{D})$ times, so the total time is $O(r_2\gamma_k^{D})$.

Finally, we have at most $O(n)$ initial task groups (c.f., Line~\ref{algo1:line3} of Algorithm~\ref{alg::1}), and by Lemma~\ref{th:subtasks}, each initial task group with seed vertex $v_i$ generates $O\left(r_1^k\right)$ sub-tasks that call Branch$(G_i,k,q,P_S,C_S,X_S)$. So, the total time cost of Algorithm~\ref{alg::1} is $O\left(nr_1^kr_2\gamma_k^{D}\right)$.
\end{proof}

\begin{figure}[t]
\centering
\includegraphics[width=\columnwidth]{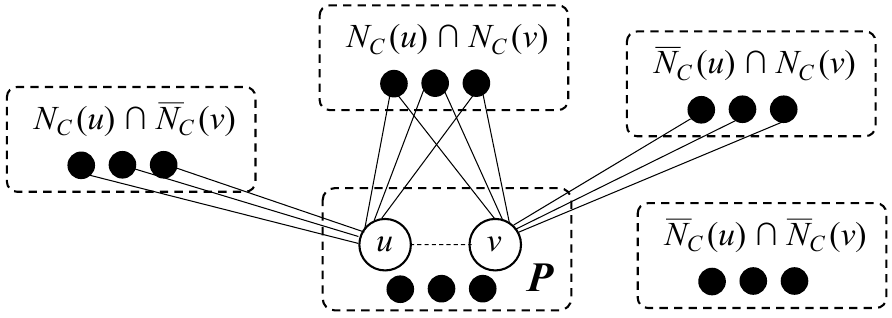}
\caption{Upper Bound Illustration for Lemma~\ref{lemma::bound1}}\label{ub1}
\end{figure}
\setlength{\textfloatsep}{5pt}

\subsection{Proof of Lemma~\ref{lemma::bound1}}\label{app:th8}
\begin{proof} To prove this, let $P_{m}\subseteq P\cup C$ be a maximum $k$-plex containing $P$. For two arbitrary vertices $u$, $v\in P$, the candidate set $C$ can be divided into four subsets as illustrated in Figure~\ref{ub1}: (1)~$N_{C}(u)\cap N_{C}(v)$, (2)~$N_{C}(u)\cap \overline{N_{C}}(v)$, (3)~$\overline{N_{C}}(u)\cap N_{C}(v)$, and (4)~$\overline{N_{C}}(u)\cap \overline{N_{C}}(v)$. Therefore,
\begin{equation*}
    \begin{split}
    |P_m|\le |P|+|N_{C}(u)\cap N_{C}(v)|+|N_{C}(u)\cap \overline{N_{C}}(v)|\\
    + |\overline{N_{C}}(u)\cap N_{C}(v)|+|\overline{N_{C}}(u)\cap \overline{N_{C}}(v)|
    \end{split}. 
\end{equation*}
Note that
\begin{gather*}
    |N_{C}(u)\cap \overline{N_{C}}(v)|+|\overline{N_{C}}(u)\cap \overline{N_{C}}(v)|=|\overline{N_{C}}(v)|\le \text{sup}_P(v),\\
    |\overline{N_{C}}(u)\cap N_{C}(v)|+|\overline{N_{C}}(u)\cap \overline{N_{C}}(v)|=|\overline{N_{C}}(u)|\le \text{sup}_P(u). 
\end{gather*}
Therefore, we have 
\begin{equation}\label{eq:ub_2sup}
|P_m|\le |P|+\text{sup}_P(u)+\text{sup}_P(v)+|N_{C}(u)\cap N_{C}(v)|,
\end{equation}
which completes the proof since $u, v\in P$ are arbitrary. \end{proof}

\subsection{Proof of Theorem~\ref{reduction1}}\label{app:th9}
\begin{proof} We prove it using Lemma~\ref{lemma::bound1}. Specifically, assume that two vertices $u_1$, $u_2\in N_{G_i}^2(v_i)$ co-occur in a $k$-plex $P$ with $|P|\geq q$, then $P$ is expanded from $P_S=\{v_i\}\cup S$, where $\{u_1, u_2\}\subseteq S\subseteq N^2_{G_i}(v_i)$ and $|S|\le k-1$ (hence $|P_S|\leq k$).

According to Eq~(\ref{eq:ub_2sup}) in Lemma~\ref{lemma::bound1}, we require
$$|P_S|+\text{sup}_{P_S}(u_1)+\text{sup}_{P_S}(u_2)+|N_{C_S}(u_1)\cap N_{C_S}(u_2)|\ge q,$$
or equivalently (recall that $|P_S|\leq k$),
$$|N_{C_S}(u_1)\cap N_{C_S}(u_2)|\ge q - k - \text{sup}_{P_S}(u_1) - \text{sup}_{P_S}(u_2).$$

If $(u_1,u_2)\in E_i$, then $\text{sup}_{P_S}(u_1)\leq k-2$ (resp.\ $\text{sup}_{P_S}(u_2)\leq k-2$), since $v_i\in P_S$ is a non-neighbor of $u_1$ (resp.\ $u_2$) besides $u_1$ (resp.\ $u_2$) itself in $P_S$. Thus, 
 \begin{equation*}
|N_{C_S}(u_1)\cap N_{C_S}(u_2)|\ge q-k-2\cdot\max\{k-2,0\}.
 \end{equation*}

While if $(u_1,u_2)\notin E_i$, then $\text{sup}_{P_S}(u_1)\leq k-3$ (resp.\ $\text{sup}_{P_S}(u_2)\leq k-3$), since $v_i\in P_S$ is a non-neighbor of $u_1$ (resp.\ $u_2$) besides $u_1, u_2\in P_S$. Thus,
\begin{equation*}
|N_{C_S}(u_1)\cap N_{C_S}(u_2)|\ge q-k-2\cdot\max\{k-3,0\}.
 \end{equation*}
 
This completes our proof of Theorem~\ref{reduction1}.\end{proof}

\subsection{Proof of Theorem~\ref{reduction3}}\label{app:th10}

\begin{proof}
Assume that $u_1\in N^2_{G_i}(v_i)$ and $u_2\in N_{G_i}(v_i)$ co-occur in a $k$-plex $P$ with $|P|\geq q$. Let us assume $u_1\in S$ and $P^+=P_S\cup\{u_2\}$, then $P^+\subseteq P$. As the proof of Theorem~\ref{reduction1} has shown, we have $|P_S|\leq k$, so $|P^+|\leq k+1$.

Applying Eq~(\ref{eq:ub_2sup}) in Lemma~\ref{lemma::bound1} with $P=P^+$, $u=u_1$ and $v=u_2$, we require
$$|P^+|+\text{sup}_{P^+}(u_1)+\text{sup}_{P^+}(u_2)+|N_{C_S^-}(u_1)\cap N_{C_S^-}(u_2)|\ge q,$$
or equivalently (recall that $|P^+|\leq k+1$),
$$|N_{C_S^-}(u_1)\cap N_{C_S^-}(u_2)|\ge q - (k+1) - \text{sup}_{P^+}(u_1) - \text{sup}_{P^+}(u_2).$$

If $(u_1,u_2)\in E_i$, then $\text{sup}_{P^+}(u_1)\leq k-2$ since $v_i\in P_S$ is a non-neighbor of $u_1$ besides $u_1$ itself in $P_S$, and $\text{sup}_{P^+}(u_2)\leq k-1$ since $u_2$ is a non-neighbor of itself in $P_S$. Thus,
\begin{eqnarray*}
|N_{C_S^-}(u_1)\cap N_{C_S^-}(u_2)| & \ge & q-(k+1)-\max\{k-2,0\}\\
& & -\ (k-1)\\
& = & q-2k-\max\{k-2, 0\}
\end{eqnarray*}

While if $(u_1,u_2)\notin E_i$, then $\text{sup}_{P^+}(u_1)\leq k-3$ since $v_i\in P_S$ is a non-neighbor of $u_1$ besides $u_1, u_2\in P^+$, and $\text{sup}_{P^+}(u_2)\leq k-2$ since $u_1,u_2\subseteq P^+$ are the non-neighbors of $u_2$. Thus, we have 
\begin{eqnarray*}
|N_{C_S^-}(u_1)\cap N_{C_S^-}(u_2)| & \ge & q-(k+1)-\max\{k-2,0\}\\
& & -\ \max\{k-3,0\}\\
& = & q-k-\max\{k-2, 0\}\\
& & -\ \max\{k-2, 1\}
\end{eqnarray*}

This completes our proof of Theorem~\ref{reduction3}.\end{proof}

\begin{figure*}[htbp]
    \centering
    \subfigure[enwiki-2021 $(k = 2,~q=40)$]{\includegraphics[width=0.24\textwidth]{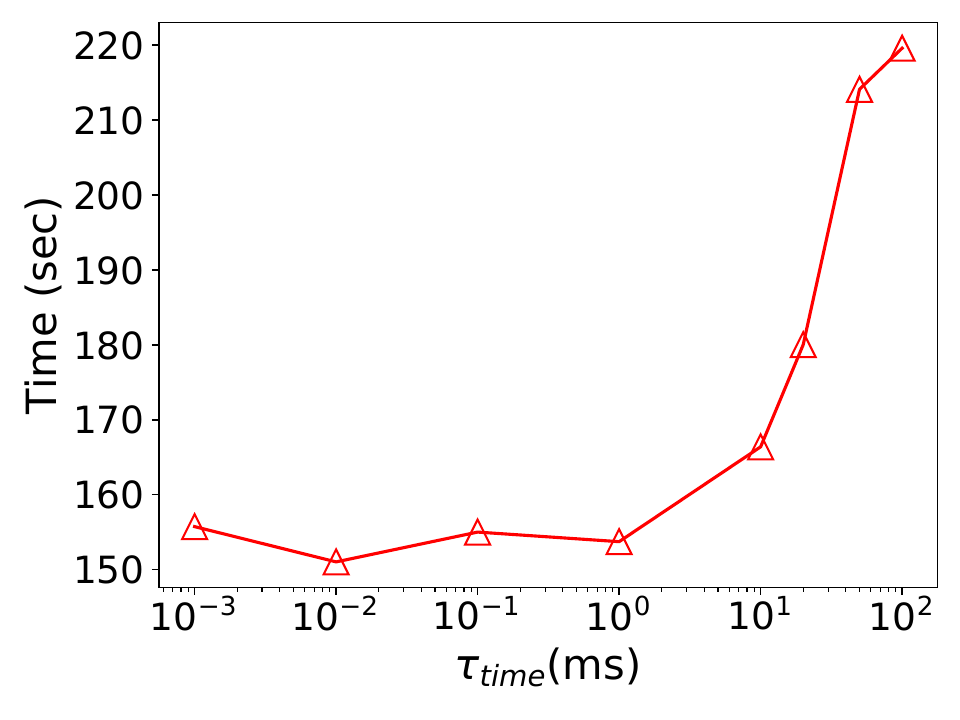}}
    \subfigure[enwiki-2021 $(k = 3,~q=50)$]{\includegraphics[width=0.24\textwidth]{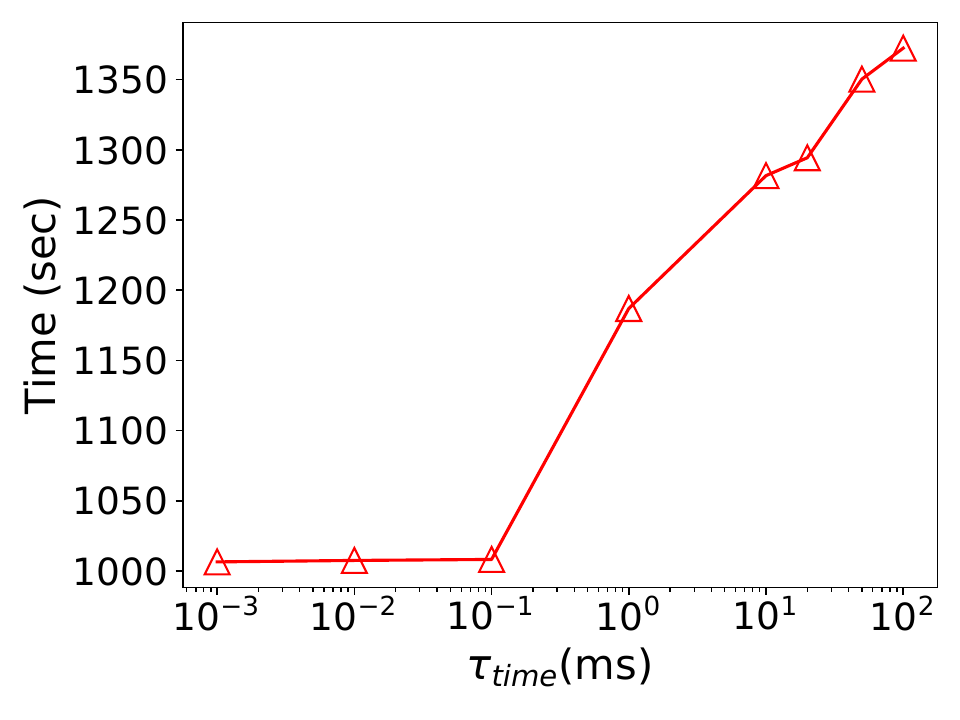}}
    \subfigure[arabic-2005 $(k = 2,~q=900)$]{\includegraphics[width=0.24\textwidth]{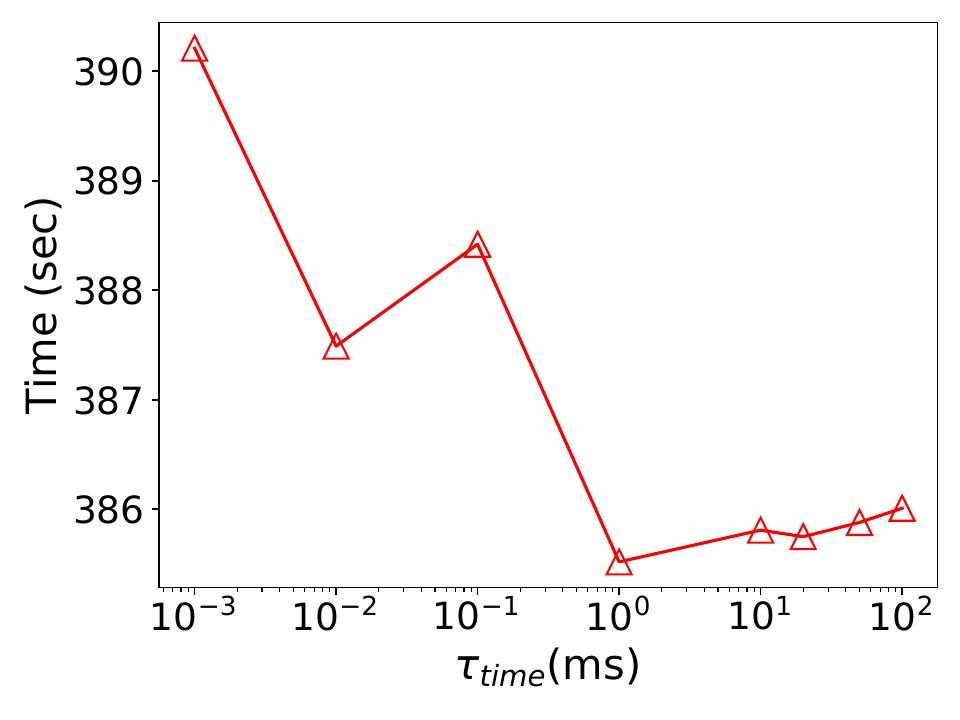}}
    \subfigure[arabic-2005 $(k = 3,~q=1000)$]{\includegraphics[width=0.24\textwidth]{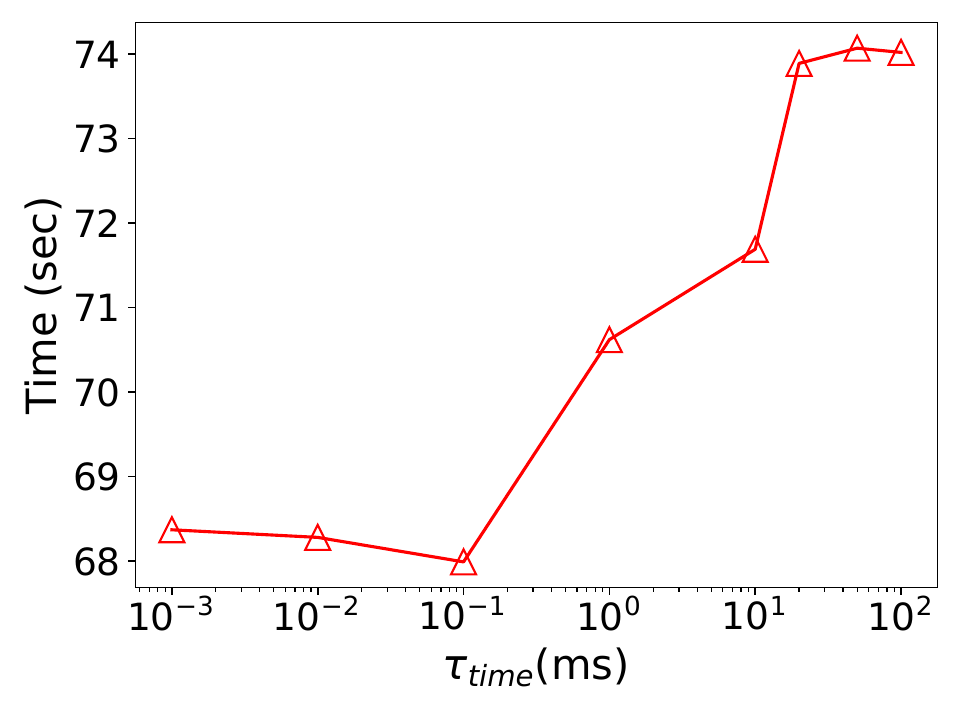}}

    \subfigure[uk-2005 $(k = 2,~q=250)$]{\includegraphics[width=0.24\textwidth]{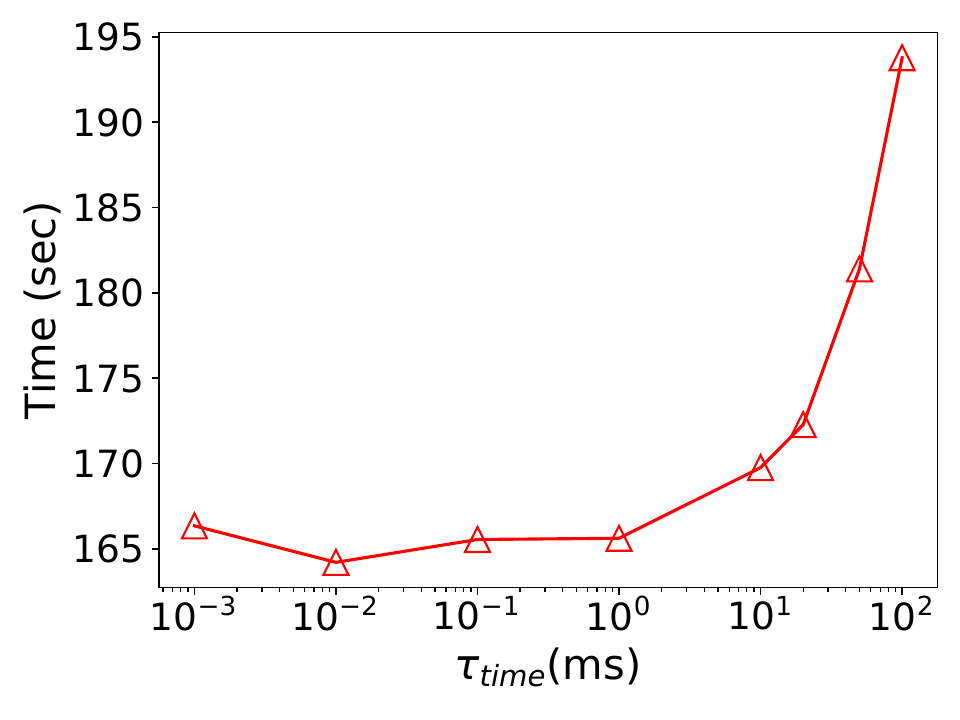}}
    \subfigure[uk-2005 $(k = 3,~q=500)$]{\includegraphics[width=0.24\textwidth]{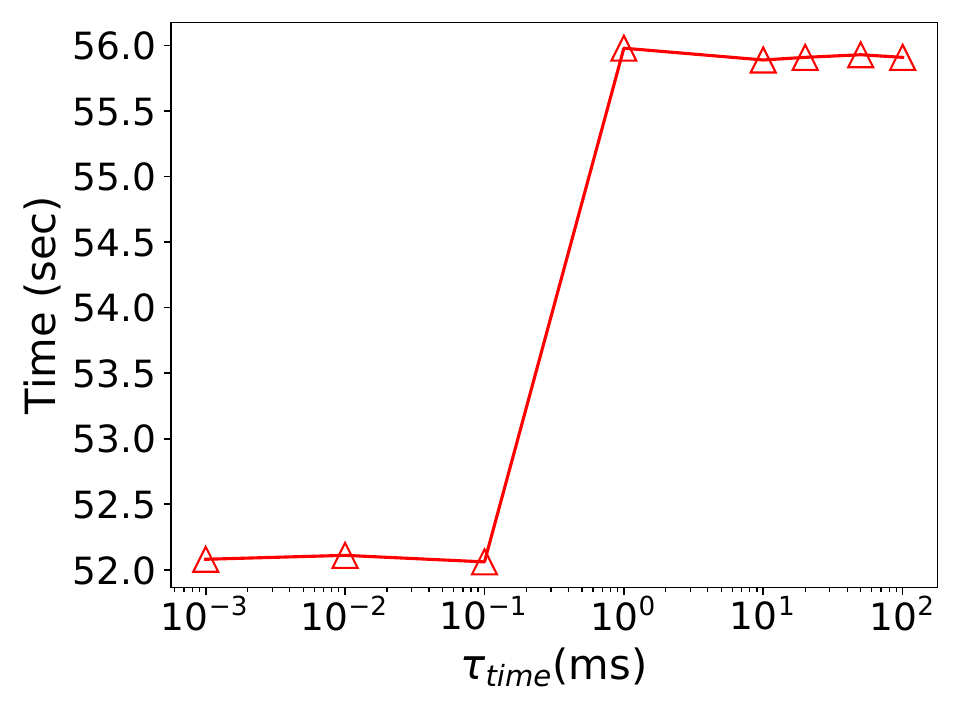}}
    \subfigure[it-2004 $(k = 2,~q=1000)$]{\includegraphics[width=0.24\textwidth]{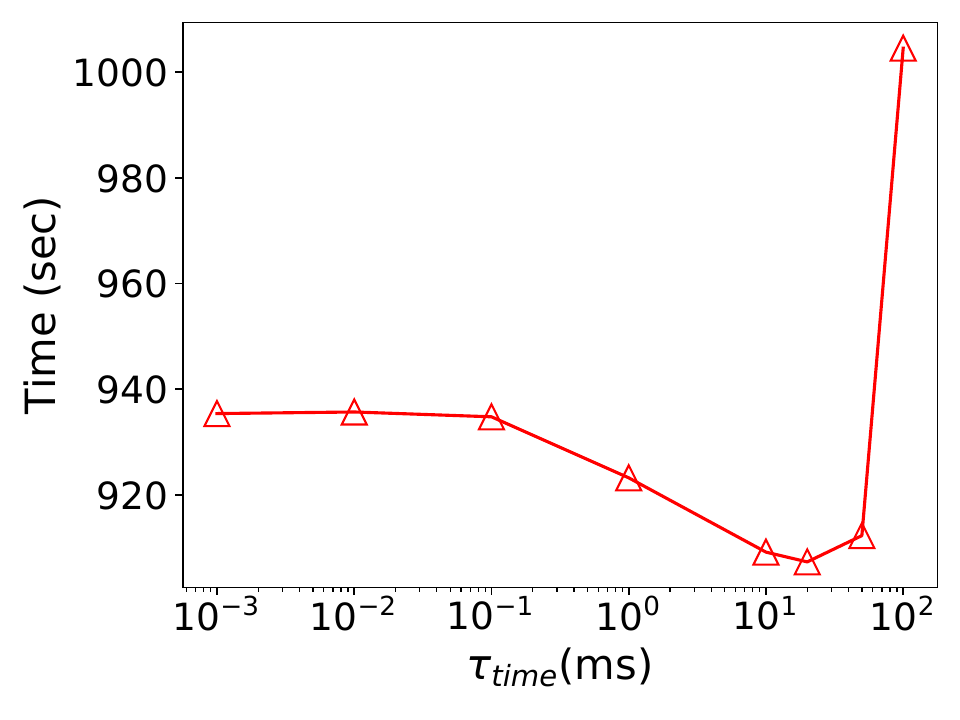}}
    \subfigure[it-2004 $(k = 3,~q=2000)$]{\includegraphics[width=0.24\textwidth]{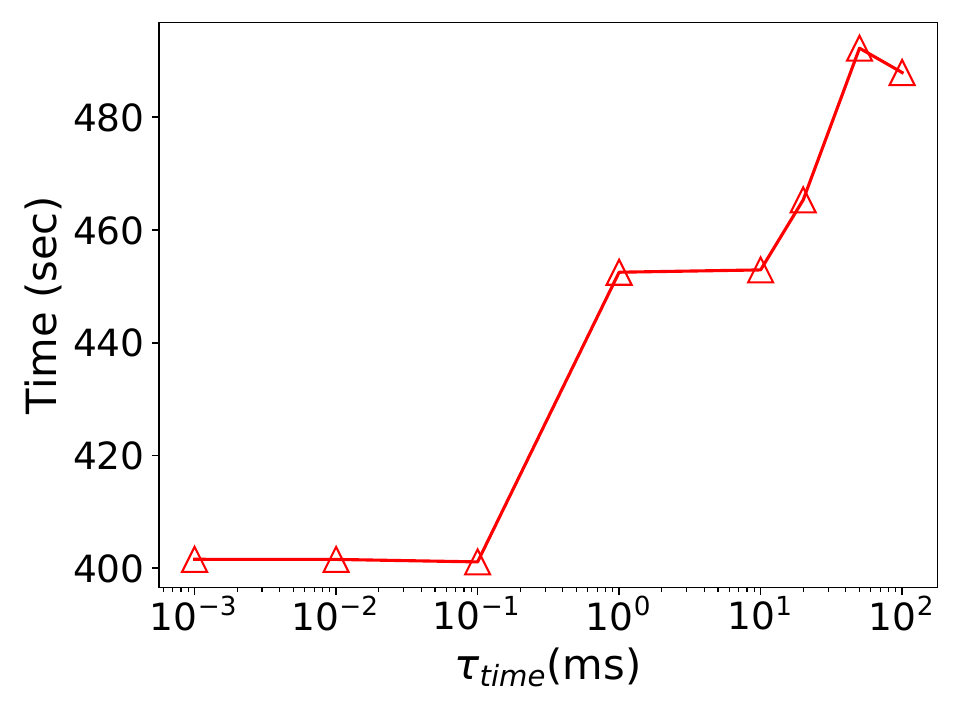}}

    \subfigure[webbase-2001 $(k = 2,~q=400)$]{\includegraphics[width=0.24\textwidth]{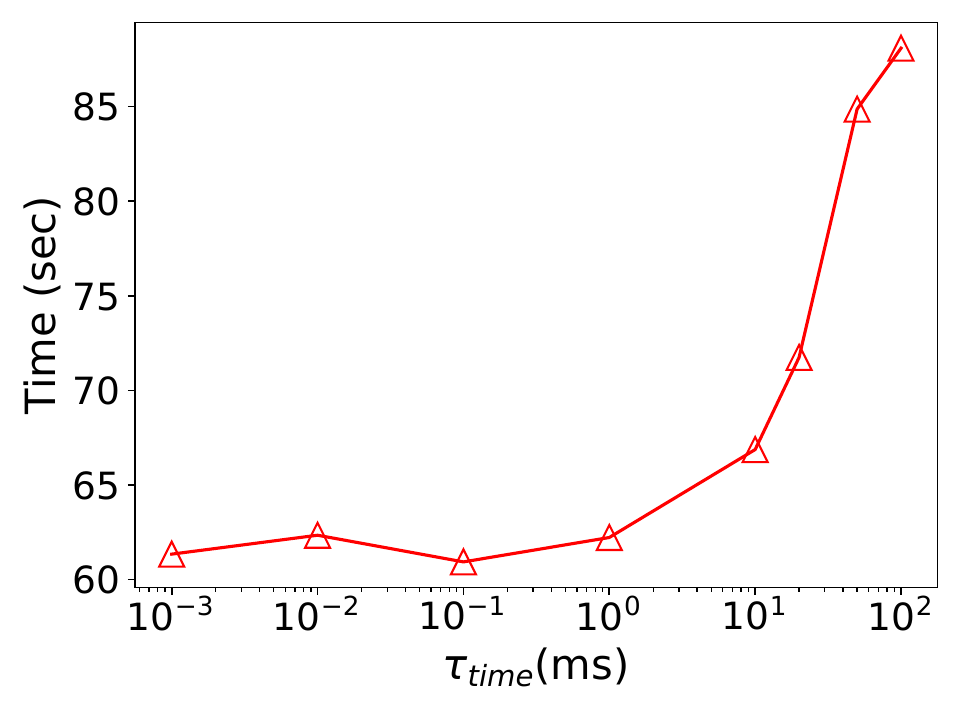}}
    \subfigure[webbase-2001 $(k = 3,~q=800)$]{\includegraphics[width=0.24\textwidth]{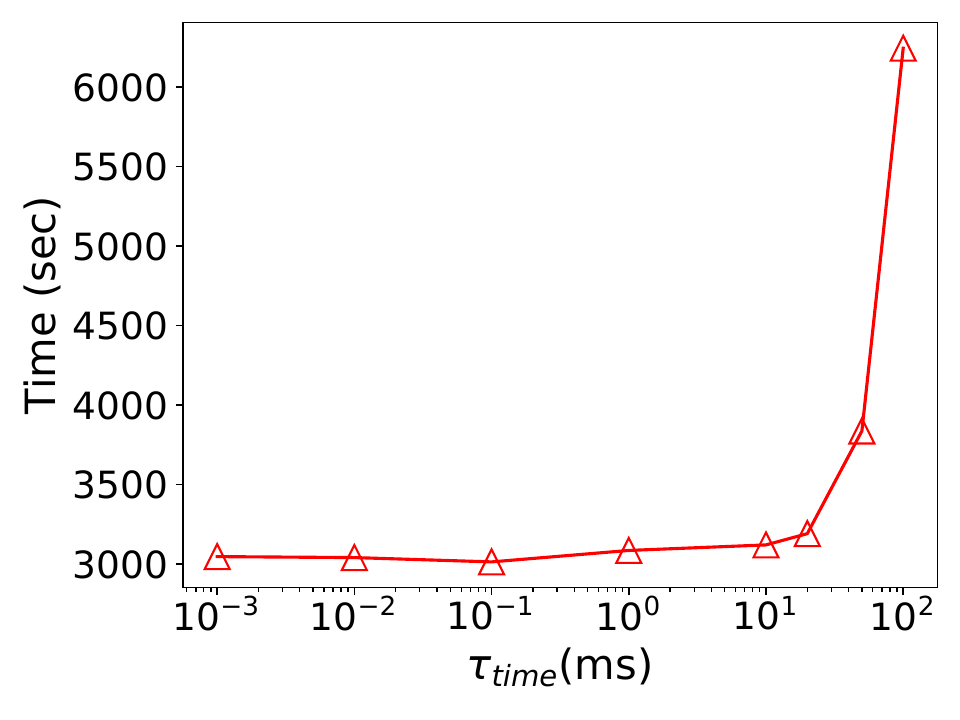}}
    \caption{The Running Time (sec) of Parallel Ours with Different $\tau_{time}$ on Five Large Datasets}
    \label{fig:different-tau}
\end{figure*}

\subsection{Proof of Theorem~\ref{reduction2}}\label{app:th11}

\begin{proof}
Assume that $u_1$, $u_2\in N_{G_i}(v_i)$ co-occur in a $k$-plex $P$ with $|P|\geq q$. Let us define $P^+=P_S\cup\{u_1, u_2\}$, then $P^+\subseteq P$. As the proof of Theorem~\ref{reduction1} has shown, we have $|P_S|\leq k$, so $|P^+|\leq k+2$.

Applying Eq~(\ref{eq:ub_2sup}) in Lemma~\ref{lemma::bound1} with $P=P^+$, $u=u_1$ and $v=u_2$, we require
$$|P^+|+\text{sup}_{P^+}(u_1)+\text{sup}_{P^+}(u_2)+|N_{C_S^-}(u_1)\cap N_{C_S^-}(u_2)|\ge q,$$
or equivalently (recall that $|P^+|\leq k+2$),
$$|N_{C_S^-}(u_1)\cap N_{C_S^-}(u_2)|\ge q - (k+2) - \text{sup}_{P^+}(u_1) - \text{sup}_{P^+}(u_2).$$

If $(u_1,u_2)\in E_i$, then $\text{sup}_{P^+}(u_1)\leq k-1$ (resp.\ $\text{sup}_{P^+}(u_2)\leq k-1$), since $u_1$ (resp.\ $u_2$) is a non-neighbor of itself in $P^+$. Thus,
\begin{equation*}
|N_{C_S^-}(u_1)\cap N_{C_S^-}(u_2)|\ge q-(k+2)-2\cdot (k-1)=q-3k.
\end{equation*}

While if $(v_1,v_2)\notin E_i$, then $\text{sup}_{P^+}(u_1)\leq k-2$ (resp.\ $\text{sup}_{P^+}(u_2)\leq k-2$), since $u_1$ (resp.\ $u_2$) is a non-neighbor of $u_1,u_2\in P^+$. Thus,
\begin{eqnarray*}
|N_{C_S^-}(u_1)\cap N_{C_S^-}(u_2)| & \ge & q-(k+2)-2\cdot\max\{k-2,0\}\\
& = & q-k-2\cdot\max\{k-1,1\}.
\end{eqnarray*}

This completes our proof of Theorem~\ref{reduction2}.\end{proof}

\section{Additional Experimental Results}
\subsection{Effect of $\tau_{time}$}\label{app:tau}
We vary $\tau_{time}$ from $10^{-3}$ to $100$ and evaluate the running time of our parallel algorithm on five large datasets with the same parameters as Table~\ref{tbl-community1}. The results are shown in Figure~\ref{fig:different-tau}, where we can see that an inappropriate parameter $\tau_{time}$ (e.g., one that is very long) can lead to a very slow performance. Note that without the timeout mechanism (as is the case in ListPlex and Ours), we are basically setting $\tau_{time}=\infty$ so the running time is expected to be longer (e.g., than when $\tau_{time}=100$) due to poor load balancing.

\subsection{Comparison of Memory Consumption}\label{app::memory}
Some illustrative results of the peak memory consumption of all three algorithms are shown in the table below. We can see that FP uses more memory on the medium-sized graph datasets, while the memory usage of ListPlex and Ours is very close. 

\renewcommand{\arraystretch}{1.3}
\begin{table}[htbp]
    \centering
	\caption{Memory Consumption of Algorithms}
        \vspace{-2mm}  
	\label{tbl-mem}
    \resizebox{0.7\columnwidth}{!}{
    \begin{tabular}{c|c|c|c|c|c}
	\toprule[1pt]
	\multirow{2}{*}{\tabincell{c}{Network}} & \multirow{2}{*}{$k$}    & \multirow{2}{*}{$q$} &  \multicolumn{3}{c}{ Memory consumption (MiB)}
	\\
	\cline{4-6}
	& &  & FP & ListPlex & Ours \\
    \hline wiki-vote & 4 & 20 & 7.04  & 19.93 & 19.72 \\
    \hline soc-epinions & 4 & 30 & 28.02 & 26.14 & 25.91 \\
    \hline email-euall & 4 & 12 & 64.27  & 34.05 & 32.76 \\
    \hline soc-pokec & 4 & 20 & 937.52  & 431.69  & 431.26 \\
    \bottomrule[1pt]
    \end{tabular}
    }
\end{table}

\begin{figure*}[t]
    \centering
    \subfigure[soc-epinions $(k = 2)$]{\includegraphics[width=0.24\textwidth]{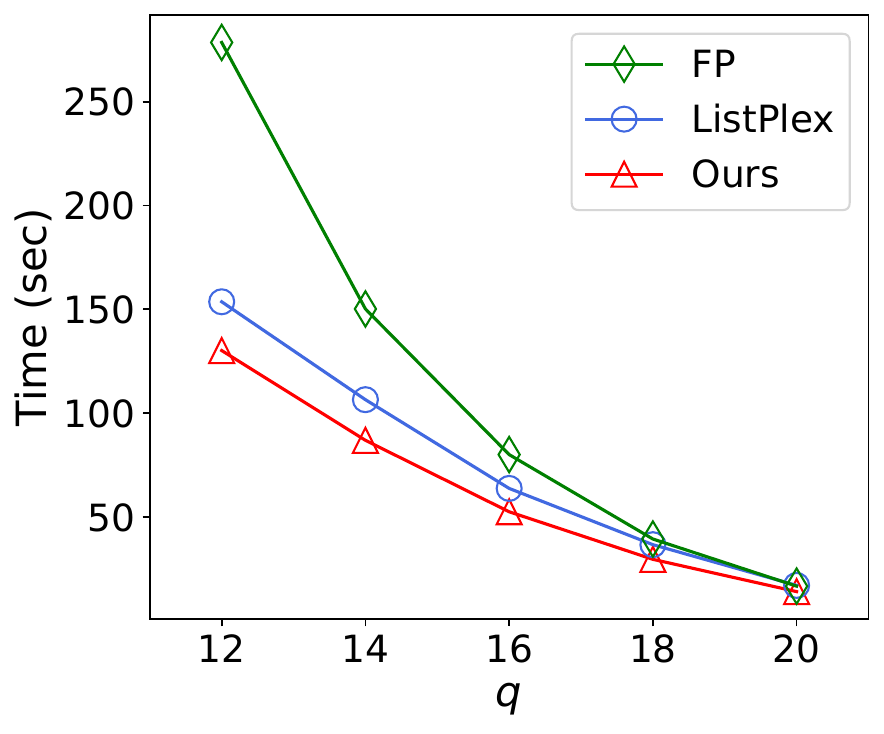}}
    \subfigure[soc-epinions $(k = 3)$]{\includegraphics[width=0.24\textwidth]{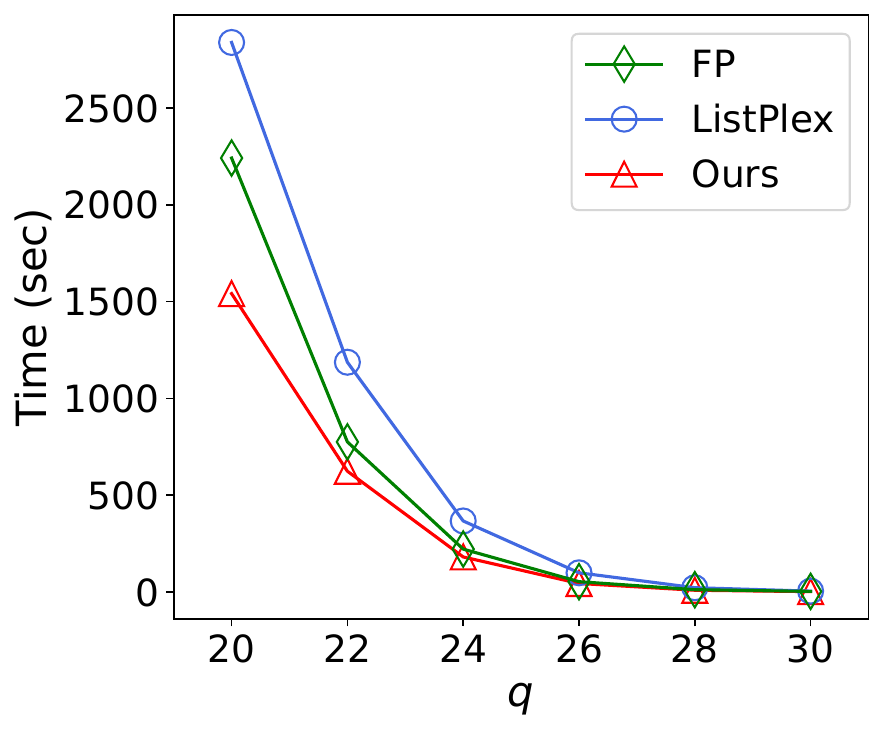}}
    \subfigure[email-euall $(k = 3)$]{\includegraphics[width=0.24\textwidth]{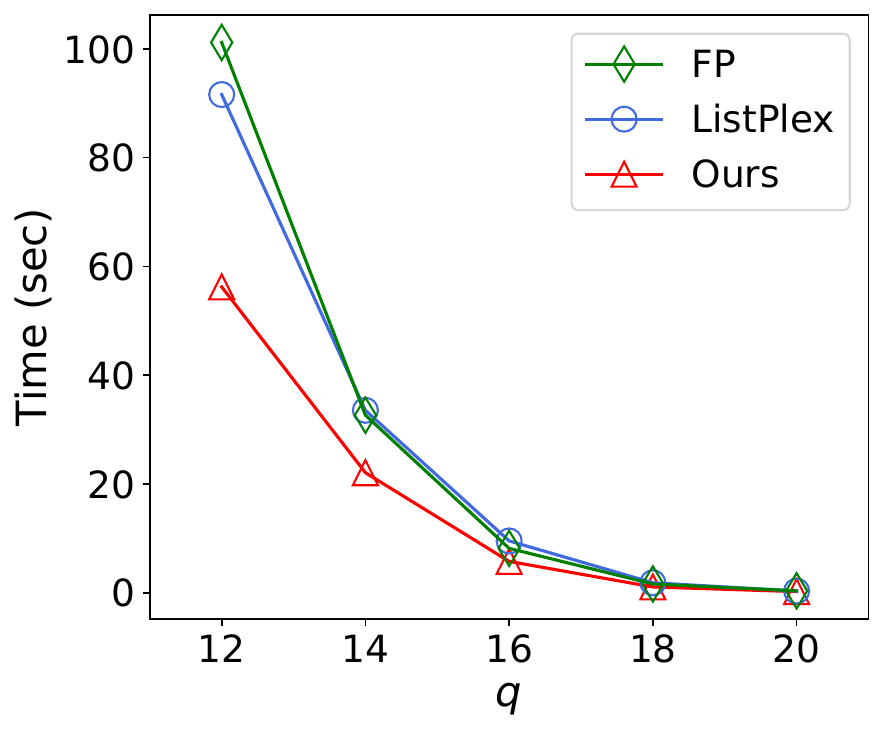}}
    \subfigure[email-euall $(k = 4)$]{\includegraphics[width=0.24\textwidth]{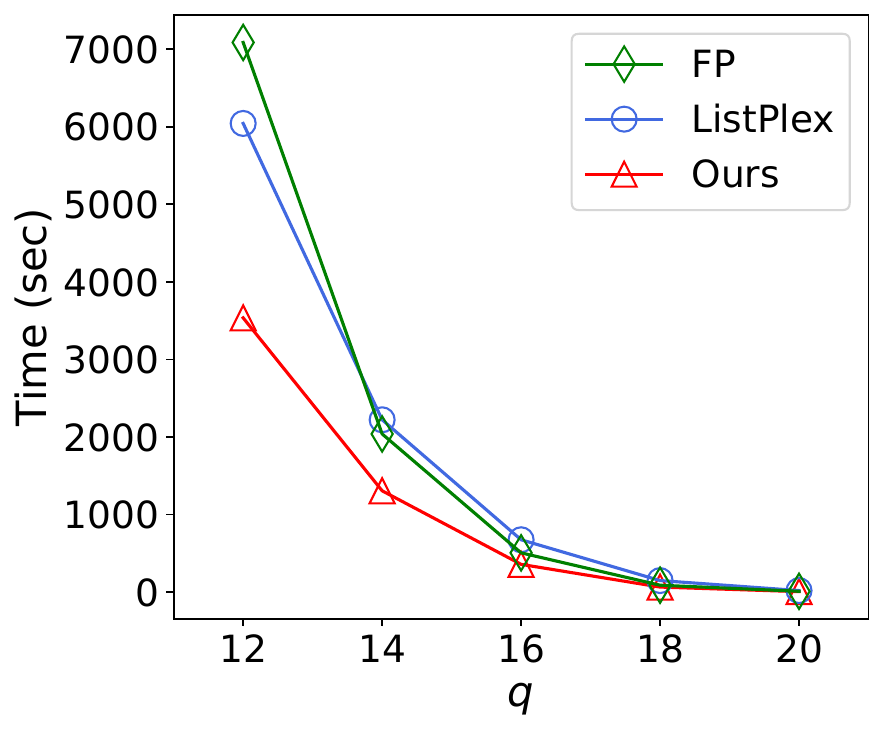}}

    \subfigure[wiki-vote $(k = 3)$]{\includegraphics[width=0.24\textwidth]{./figure/var_q_4}}
    \subfigure[wiki-vote $(k = 4)$]{\includegraphics[width=0.24\textwidth]{./figure/var_q_5}}
    \subfigure[soc-pokec $(k = 3)$]{\includegraphics[width=0.24\textwidth]{./figure/var_q_6}}
    \subfigure[soc-pokec $(k = 4)$]{\includegraphics[width=0.24\textwidth]{./figure/var_q_7}}
    \vspace{-5pt}
    \caption{Running Time (sec) of the Three Algorithms on Various Datasets and Parameters}
    \label{fig:different-q-total}
\end{figure*}

\begin{figure*}[t]
    \centering
    \subfigure[soc-epinions $(k = 2)$]{\includegraphics[width=0.22\textwidth]{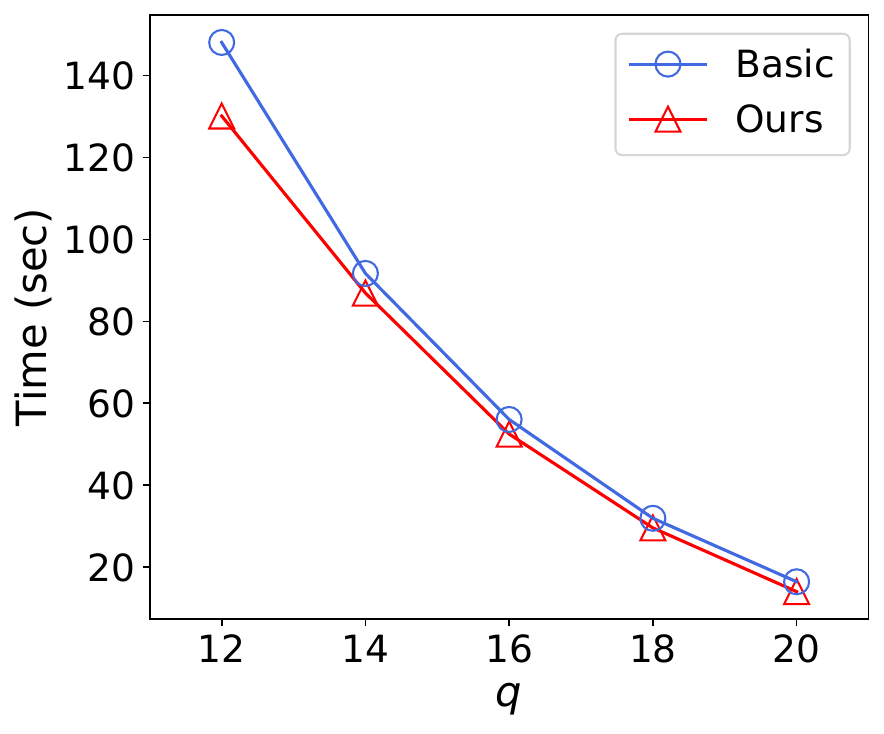}}
    \subfigure[soc-epinions $(k = 3)$]{\includegraphics[width=0.22\textwidth]{./figure/var_q_ablation_1}}
    \subfigure[email-euall $(k = 3)$]{\includegraphics[width=0.22\textwidth]{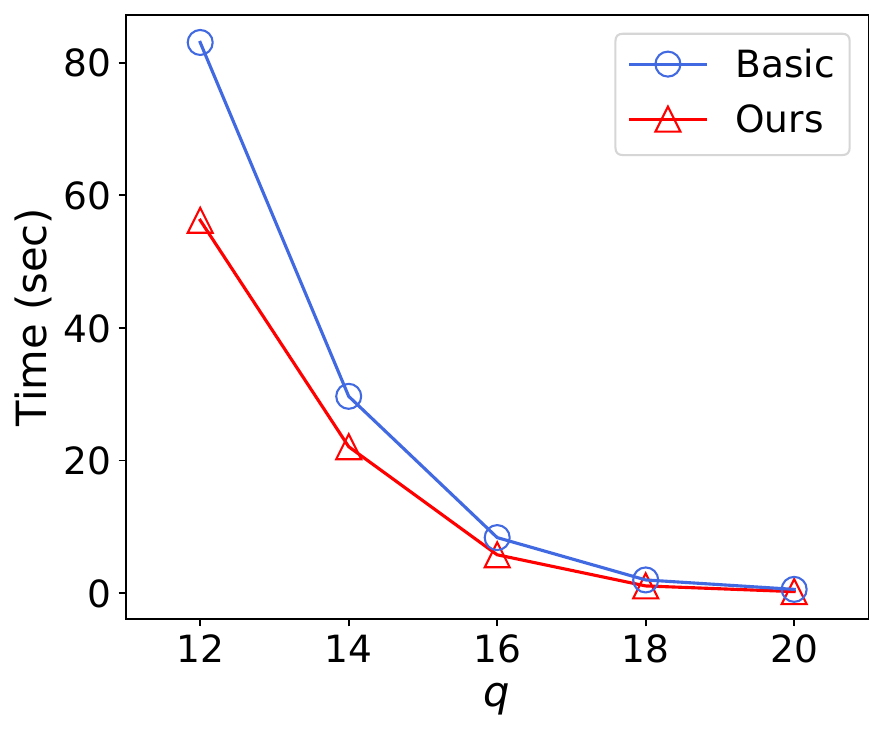}}
    \subfigure[email-euall $(k = 4)$]{\includegraphics[width=0.22\textwidth]{./figure/var_q_ablation_2}}

    \subfigure[wiki-vote $(k = 3)$]{\includegraphics[width=0.22\textwidth]{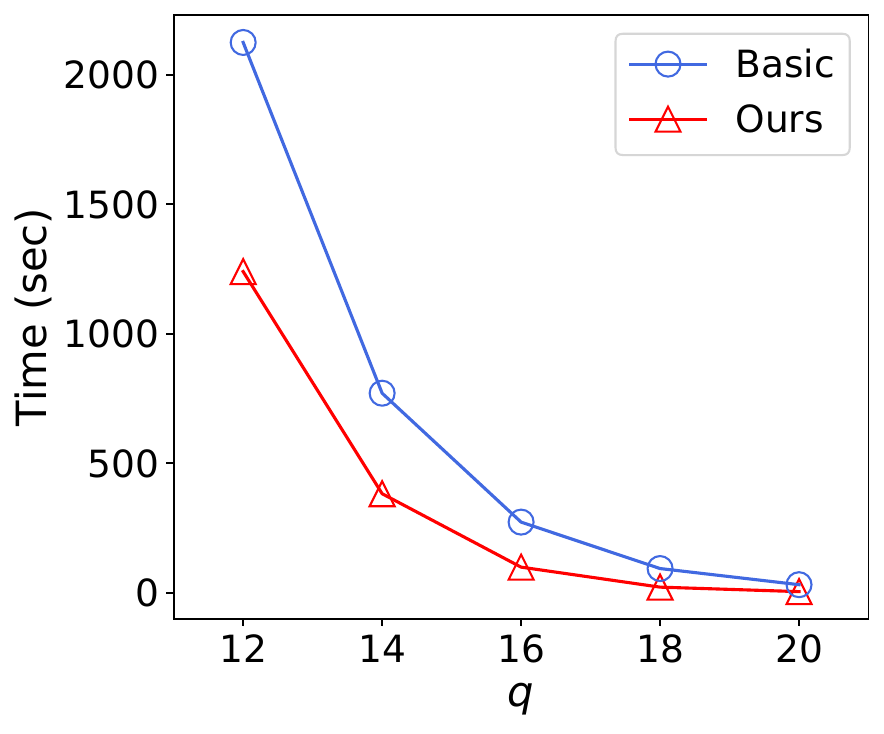}}
    \subfigure[wiki-vote $(k = 4)$]{\includegraphics[width=0.22\textwidth]{./figure/var_q_ablation_5}}
    \subfigure[soc-pokec $(k = 3)$]{\includegraphics[width=0.22\textwidth]{./figure/var_q_ablation_6}}
    \subfigure[soc-pokec $(k = 4)$]{\includegraphics[width=0.22\textwidth]{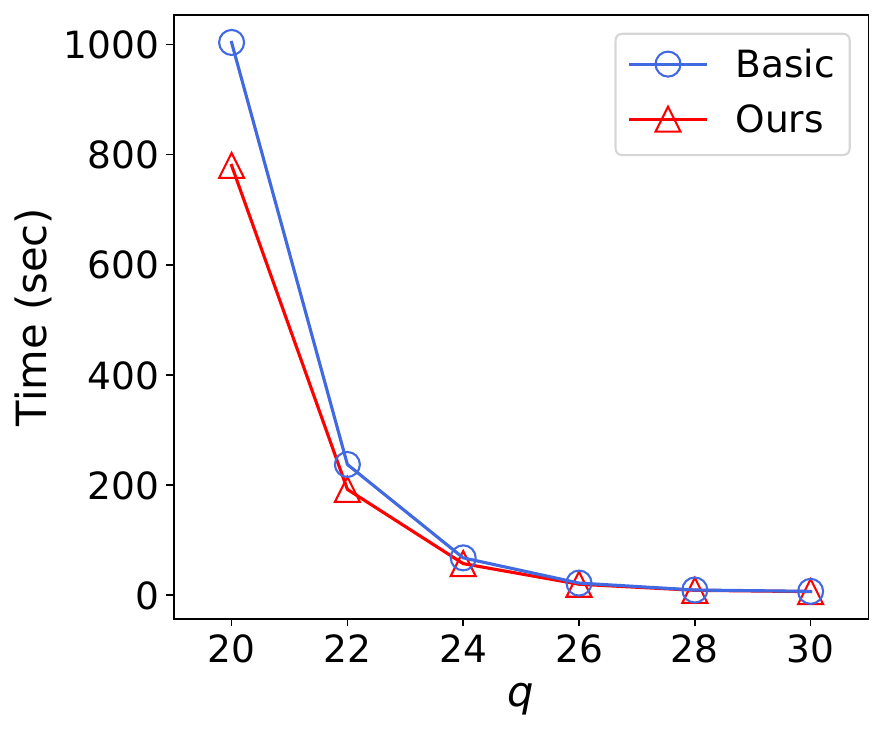}}
    \vspace{-5pt}
    \caption{Running time (sec) of Our Basic and Optimized Algorithms on Various Datasets and Parameters}
    \label{fig:different-q-ablation-total}
\end{figure*}

\subsection{Effect of $q$}\label{app::different-q-ablation-total}
Figure~\ref{fig:different-q-total} shows how the performance of sequential algorithms changes when $q$ varies. In each subfigure, the horizontal axis is $q$, and the vertical axis is the total running time. As Figure~\ref{fig:different-q-total} shows, Ours (red line) consistently uses less time than ListPlex and FP. For example, Ours is $4 \times$ faster than ListPlex on wiki-vote when $k = 4$, $q = 20$.

\subsection{Ablation Study: Basic v.s.\ Ours}\label{app::basic}
Figure~\ref{fig:different-q-ablation-total} compares the running time between Basic and Ours as $k$ and $q$ vary. We can see that Ours is consistently faster than the basic version with different $k$ and $q$. This demonstrates the effectiveness of our pruning rules. 

\end{appendix}

\end{document}